\RequirePackage[l2tabu, orthodox]{nag}		
\documentclass[10pt]{article}
\usepackage[T1]{fontenc}
\usepackage[frenchmath]{mathastext}							
\usepackage{amsmath}                        
\numberwithin{equation}{section}
\usepackage{graphicx} 											
\usepackage{adjustbox,float}
\graphicspath{{Figs/}}
\usepackage{enumitem}												
\usepackage{mdwlist}												
\usepackage[dvipsnames]{xcolor}							
\usepackage[plainpages=false, pdfpagelabels]{hyperref} 
\hypersetup{
colorlinks   = true,
citecolor    = RoyalBlue, 
linkcolor    = RubineRed, 
urlcolor     = Turquoise
}

\usepackage{amssymb}                        
\usepackage{mathrsfs}						
\usepackage{dsfont}							
\usepackage[paperwidth=8.5in,paperheight=11in,top=1.25in, bottom=1.25in, left=1.0in, right=1.0in]{geometry} 
\usepackage{mathtools}                      
\mathtoolsset{showonlyrefs=true}          
\usepackage{amsthm}                         
\allowdisplaybreaks                       
\usepackage{algorithm}
\usepackage[noend]{algpseudocode}

\theoremstyle{plain}
\newtheorem{theorem}{Theorem}[section]
\newtheorem{lemma}[theorem]{Lemma}
\newtheorem{proposition}[theorem]{Proposition}
\newtheorem{corollary}[theorem]{Corollary}
\theoremstyle{definition}
\newtheorem{definition}[theorem]{Definition}
\newtheorem{example}[theorem]{Example}
\newtheorem{remark}[theorem]{Remark}
\newtheorem{assumption}[theorem]{Assumption}
\newtheorem{problem}[theorem]{Problem}

\newcommand\Cb{\mathds{C}}
\newcommand\Eb{\mathds{E}}
\newcommand\Fb{\mathfrak{F}}
\newcommand\Ib{\mathds{1}}

\newcommand\Pb{\mathds{P}}

\newcommand\Rb{\mathds{R}}

\newcommand\Nb{\mathbb{N}}

\newcommand\Ac{\mathcal{A}}
\newcommand\Bc{\mathscr{B}}
\newcommand\Cc{\mathcal{C}}
\newcommand\Dc{\mathscr{D}}
\newcommand\Ec{\mathscr{E}}
\newcommand\Fc{\mathscr{F}}
\newcommand\Gc{\mathscr{G}}

\newcommand\Lc{\mathscr{L}}

\newcommand\Sc{\mathscr{S}}

\newcommand\Jc{\mathcal{J}}
\newcommand\Ic{\mathcal{I}}
\newcommand\Uc{\mathcal{U}}

\newcommand\Xc{\mathcal{X}}

\newcommand\eps{\varepsilon}
\newcommand\om{\omega}
\newcommand\Om{\Omega}
\newcommand\sig{\sigma}
\newcommand\Sig{\Sigma}

\newcommand\gam{\gamma}

\newcommand\del{\delta}

\newcommand\Bv{\textbf{B}}
\newcommand\Gv{\textbf{G}} 

\newcommand\Sv{\textbf{S}}

\newcommand\gv{\mathbf{g}}

\newcommand\Xv{\mathbf{X}}

\newcommand\sv{\mathbf{s}}

\newcommand\thetav{{\boldsymbol\theta}}
\newcommand\Thetav{{\boldsymbol\Theta}}

\newcommand\Lamv{{\boldsymbol\Lambda}}
\newcommand\lamv{{\boldsymbol\lambda}}

\newcommand\fh{\hat{f}}

\newcommand\muh{\hat{\mu}}

\newcommand\Act{\tilde{\Ac}}

\newcommand\ft{\tilde{f}}
\newcommand\gt{\tilde{g}}

\newcommand\It{\tilde{I}}
\newcommand\Jt{\tilde{J}}
\newcommand\mt{\tilde{m}}

\newcommand\Ut{\widetilde{U}}

\newcommand\nut{\widetilde{\nu}}
\newcommand\mut{\tilde{\mu}}

\newcommand{\ceil}[1]{\lceil{#1}\rceil}
\newcommand{\floor}[1]{\lfloor{#1}\rfloor}

\newcommand\ii{\mathtt{i}}
\newcommand\dd{\mathrm{d}}
\newcommand\ee{\mathrm{e}}

\newcommand{\as}{\text{a.s.}}
\newcommand{\CMIM}{{CMIM}}

\newcommand{\1}{{\mathds{1}}}

\newcommand{\diag}{\operatorname{diag}}

\newcommand{\supp}{\operatorname{supp}}
\newcommand{\conv}{\boldsymbol{*}}
\newcommand{\F}{\mathbb{F}}
\newcommand \al {\alpha}

\providecommand{\keywords}[1]{\noindent\textbf{\textit{Keywords: }} #1}

\begin{document}

\title{Predictable Forward Performance Processes\\ in Complete Markets}

\author{
Bahman Angoshtari
\thanks{Department of Mathematics, University of Miami, Coral Gables, FL 33146, USA \textbf{e-mail}: \url{bangoshtari@miami.edu}}
}

\date{This version: \today}

\maketitle

\begin{abstract}
We establish existence of Predictable Forward Performance Processes (PFPPs) in conditionally complete markets, which has been previously shown only in the binomial setting. Our market model can be a discrete-time or a continuous-time model, and the investment horizon can be finite or infinite. We show that the main step in construction of PFPPs is solving a one-period problem involving an integral equation, which is the counterpart of the functional equation found in the binomial case. Although this integral equation has been partially studied in the existing literature, we provide a new solution method using the Fourier transform for tempered distributions. We also provide closed-form solutions for PFPPs with inverse marginal functions that are completely monotonic and establish uniqueness of PFPPs within this class. We apply our results to two special cases. The first one is the binomial market and is included to relate our work to the existing literature. The second example considers a generalized Black-Scholes model which, to the best of our knowledge, is a new result.
\end{abstract}


%
%

\section{Introduction}\label{sec:intro}

In the classical approach to portfolio choice, one assumes that a market model for the entire investment period is known and that the investor's risk preferences over the investment period are pre-specified exogenously to the market. Despite its mathematical foundations and theoretical appeal, this approach has several shortcomings. Its most unrealistic assumption is, perhaps, its pre-commitment to a market model (including specific values for the model parameters) for the \emph{entire} investment horizon. In reality, portfolio managers believe in models for time periods that are far shorter than their (perceived) investment horizons. This is why models are frequently calibrated during the investment horizon and not just once at the beginning. It is thus more realistic to think of the investment horizon as a sequence of shorter ``\emph{calibration}'' periods. At the beginning of each calibration period, a model is calibrated (using, say, historical data and/or expert opinion) and the portfolio manager has confidence in the calibrated model until the end of the calibration period. 

The objective of our paper is to develop an investment paradigm that can be applied to the following scenario. We assume that model calibration is performed at times $0, 1, 2, \dots$.\footnote{More generally, we could have assumed that model calibration occurs at random times $0=\tau_0<\tau_1<\tau_2<...$ such that $\tau_n$ is known at time $\tau_{n-1}$. To ease the notation, we have taken $\tau_n=n$.} In other words, we assume that the portfolio manager calibrate a market model at time $0$ and commit to it over the time period $[0,1]$. She will re-calibrate the model at time $1$ and commit to the new model over the period $[1,2]$, and so on. Note two things at the outset. Firstly, we have not chosen an investment horizon. The investment horizon can be  deterministic, stochastic, or infinite. Secondly, we have not chosen a trading frequency. Trading can be done in discrete-time and as frequent as model calibration (i.e. at time $0,1,2,\dots$); it can be done in discrete time but more frequently than model calibration (e.g. at times $0,1/N, 2/N,\dots$ for some positive integer $N$); or trading can be done in continuous time (i.e. at any $t\ge0$).

What separates the above setting from the classical approach is that we don't pre-commit to a single market model. In particular, we \emph{do not} model the calibration procedure. For instance, we may assume that, during period $[n-1,n]$, the stock price follows the Black-Scholes model with drift $\mu_n$ and volatility $\sig_n$, in which $\mu_n$ and $\sig_n$ are random variables that are known at time $n-1$. However, we do not model how $\mu_n$ and $\sig_n$ evolve through time. Note also that our setting is not as general as the one in the literature on model ambiguity. In particular we commit to a class of models (say, Black-Scholes or binomial), although we do not pre-commit (i.e. at time $0$) to specific values of the model parameters.

Our proposed investment paradigm is based on the idea of \textit{forward performance measurement}, which was proposed and extended in a series of papers by Musiela and Zariphopoulou, see \cite{MZ09, MZ10, MZ11}. The literature of forward performance measurement has since grown significantly and we refer to \cite{HeStrubZariphopoulou2021} for a recent account of related work. The main idea of the forward approach is that instead of fixing, as in the classical setting, an investment horizon, a market model and a terminal utility, one starts with an initial  performance measurement and updates it \textit{forward in time} as the market and other underlying stochastic factors evolve. The evolution of the forward process is dictated by a forward-in-time version of the dynamic programming principle and, thus, it ensures time-consistency across all different times.

In most forward performance measurement models, the investor's preference is updated continuously in time. We, however, seek an investment paradigm in which the investor preference is updated at discrete times $0,1,2,\dots$ (i.e. when the model is calibrated).

In particular, we will develop an investment framework according to the following forward-in-time iterative procedure. Initially, the portfolio manager's preference toward her initial wealth is (exogenously) given by a utility function $U_0$. At time $0$, a market model is calibrated for the time period $[0,1]$. Let the model be parameterized by some parameters $\thetav_1$ (say, the stock drift and volatility over the period $[0,1]$), so that the outcome of calibration is observing the value of $\thetav_1$. Still at time $0$, a utility function $U_1$ for wealth at time $1$ is chosen that is consistent with the utility function $U_0$. By being consistent, we mean that $U_1$ satisfies
\begin{align}\label{eq:InvMerton_1}
	U_0(x) = \sup_{X_1\in \Ac_1(x)} \Eb_{\thetav_1}[U_1(X_1)];\quad x>0,
\end{align}
in which $\Ac_1(x)$ is the set of all admissible wealth $X_1$ at time 1 starting with initial wealth $x$ at time 0, and $\Eb_{\thetav_1}[\cdot]$ is the expectation operator under the calibrated model. Note that \eqref{eq:InvMerton_1} is the inverse problem of the classical Merton problem, in that the value function $U_0$ is known while the terminal utility function $U_1$ is unknown. Note also that $U_1$ depends on $\thetav_1$ (i.e. the calibrated model) through the expectation operator $\Eb_{\thetav_1}[\cdot]$. In particular, $U_1$ is in the form $U_1(\cdot,\thetav_1)$. Having identified a market model and a terminal utility $U_1$, we may use the classical approach to invest optimally over the time period $[0,1]$.

We repeat this procedure for the second calibration period $[1,2]$. At time $1$, we know the utility $U_1(\cdot,\thetav_1)$, and re-calibrate the model to obtain $\thetav_2$ (say, the stock drift and volatility over the period $[1,2]$). Still at time 1, we choose a utility function $U_2$ for wealth at time 2 that is consistent with $U_1(\cdot,\thetav_1)$. In other words, $U_2$ solves
\begin{align}\label{eq:InvMerton_2}
	U_1(x,\thetav_1) = \sup_{X_2\in \Ac_2(x)} \Eb_{\thetav_2}[U_2(X_2)];\quad x>0,
\end{align}
in which $\Ac_2(x)$ is the set of all admissible wealth $X_2$ at time 2 starting with initial wealth $x$ at time 1, and $\Eb_{\thetav_2}[\cdot]$ is the expectation operator under the re-calibrated model. Now, $U_2$ takes the form $U_2\big(\cdot,(\thetav_1,\thetav_2)\big)$, since it depends on $\thetav_1$ because of $U_1$, and on $\thetav_2$ because of $\Eb_{\thetav_2}[\cdot]$. With a market model and a terminal utility $U_2$ at hand, we may use the classical approach to invest optimally over the time period $[1,2]$. We can continue this procedure indefinitely.

Motivated by the above procedure, \cite{AZZ20} proposed a new forward performance measurement model, called \emph{Predictable Forward Performance Process (henceforth, PFPP)}, in which the investor's preferences are endogenous and predictable with regards to an underlying market information set and, furthermore, are updated at discrete times. Although they provided a rather general definition, their analysis and results only applied to a binomial model in which trading occurs as frequently as performance updates (i.e. the market setting of Example \ref{ex:Binom} below with $N=1$). In this setting, they found that the key step in the construction of PFPPs is to solve a one-period {\it inverse} Merton problem, namely,
\begin{align}\label{eq:InvMerton_intro}
	U_0(x)=\underset{X\in \Ac_n(x)}{\sup}\Eb[U_1(X)],\quad x>0,
\end{align}
in which $\Ac_n(x)$ denotes the set of admissible wealth at $n$ starting with wealth $x$ at $n-1$, $U_0$ is known, and $U_1$ is to be found. They showed that, in their binomial setting, \eqref{eq:InvMerton_intro} reduces to the linear functional equation
\begin{align}\label{eq:FnEQ_AZZ}
	I_{1}(ay)+bI_{1}(y)=(1+b)\,I_{0}(c\,y);\quad y>0,
\end{align}
in which $a,b,c>0$ are known constants (determined by the binomial parameters), $I_0$ is a given inverse marginal function (see \eqref{eq:Ic} below), and $I_1$ is an unknown inverse marginal function to be determined. They established conditions for the existence and uniqueness for the solution of \eqref{eq:FnEQ_AZZ}. \cite{LiangStrubWang2021} extended the model by providing existence of PFPPs in a binomial market in which trading is more frequent than performance evaluation (i.e. the market setting of Example \ref{ex:Binom} below with $N\ge2$). They applied their result to find the optimal policy taken by a robo-advisor. \cite{StrubZhou2021} considered a complete semi-martingale market model, but, focused mainly on the one-period inverse Merton problem \ref{eq:InvMerton_intro}. They showed that the counterpart of \eqref{eq:FnEQ_AZZ} is the following integral equation,
\begin{align}\label{eq:IntEq_intro}
	\int_{\Rb_+} \rho I_1(y\rho)\dd \nu(\rho) = I_0(y);\quad y>0,
\end{align}
in which $\nu$ is a probability measure on $\Rb_+$, $I_0$ is a given inverse marginal function, and $I_1$ is an unknown inverse marginal. They pay special attention to the behavior of  Arrow–Pratt measure of risk-tolerance of the pair $(U_0,U_1)$ in \eqref{eq:InvMerton_intro}, and characterized the
class of solutions with time invariant risk-tolerance. They also provided explicit solution for $\eqref{eq:InvMerton_intro}$ assuming CRRA and SAHARA utility function. To the best of our knowledge, there are no further work on PFPPs.

In this paper, we consider PFPPs in a general complete market. Our contribution to the existing literature is threefold. Firstly, in Theorem \ref{thm:verif}, we provide a set of conditions for existence of PFPPs in complete markets. To the best of our knowledge, there has not been such a result beyond those provided by \cite{AZZ20} and \cite{LiangStrubWang2021} in the binomial model. By existence of PFPPs, we mean conditions for existence of a PFPP in a setting with multiple evaluation periods. Note that \cite{StrubZhou2021} mainly considered the single period inverse Merton problem \eqref{eq:InvMerton_intro} and, as they explicitly mentioned on two occasions, did not provide any multi-period existence condition for PFPPs in their general model.\footnote{In particular, the second paragraph on page 333 of \cite{StrubZhou2021} mentions that ``Showing the existence of a discrete-time predictable forward process in the general setting and constructing such processes by sequentially solving the associated generalised integral equations and showing that their solutions are predictable all remain challenging open problems not addressed in this paper.'' Furthermore, on the last paragraph of page 336 therein, it is mentioned that ``we neither derive results on existence and uniqueness of solutions to (2.2), nor do we provide conditions for the required measurability of the solution in case it exists.''} Our existence conditions reduces construction of PFPPs into an iterative procedure whose main step is a single period problem, namely, Problem \ref{prob:IntEq} below. This one-period problem is the counterpart of the inverse Merton problem \eqref{eq:InvMerton_intro} in our setup.

As our second contribution, in Proposition \ref{prop:ExistenceUniquness}, we provide a general method for solving the integral equation \eqref{eq:IntEq} that appears in Problem \ref{prob:IntEq}. This method first turns the integral equation into a deconvolution problem (namely, \eqref{eq:Deconvolution} below) and then applies Fourier analysis to solve it. To the best of our knowledge, we are the first to use the Fourier transform for constructing PFPPs. Our arguments, however, are in the same spirit as those used by \cite{Kallblad2020}, who applied Weierstrass transform for solving the single-period inverse Merton problem \eqref{eq:InvMerton_intro} in the Black-Scholes market model. The assumptions and arguments used for Proposition \ref{prop:ExistenceUniquness} are rather technical and rely on the theory of Fourier transform for tempered distributions.

Our third contributions is Theorem \ref{thm:CMIM} which provides a closed-form solution for Problem \ref{prob:IntEq} assuming that the initial inverse marginal function is completely monotonic. The theorem also establishes the uniqueness of the solution within the class of \emph{completely monotonic inverse marginal (CMIM)} functions. See \cite{Kallblad2020} and \cite{MostovyiSirbuZariphopoulou2020} for further discussion on CMIM functions. To the best of our knowledge, we are the first to consider CMIMs in constructing PFPPs. We should mention, however, that \cite{AZZ20} provided closed-form solution for CRRA utilities while \cite{StrubZhou2021} provided closed-form solution for SAHARA utility, which are special cases of our result.

Finally, by combining our first and third results, we provide an explicit investment procedure using the framework provided by PFPPs. See Theorem \ref{thm:PFPP_CMIM} and Algorithm \ref{alg:PFPP_CMIM}.

We have included two examples to illustrate our results. The first one considers the binomial market and is included to relate our work to the existing literature. In the second example, we construct PFPPs in a generalized Black-Scholes market which, to the best of our knowledge, is a new result.

The rest of the paper is organized as follows. Subsection \ref{sec:notations} includes frequently used notations. 
In Section \ref{sec:Market}, we set up the market model and state our main standing assumption, namely, Assumption \ref{asum:rho}. In Section \ref{sec:PFPP}, we define PFPPs (see Definition \ref{def:PFPP}) and provide conditions for their existence in Theorem \ref{thm:verif}. We also discuss the iterative construction of PFPPs (see Subsection \ref{sec:PFPP_construction}). 
In Section \ref{sec:ConvEq}, we first show that the main step of the construction procedure is solving a single-period problem, namely, Problem \ref{prob:IntEq}. Then, in Subsection \ref{sec:deconvolution}, we solve the integral equation in Problem \ref{prob:IntEq} by applying the Fourier transform. Subsection \ref{sec:CMIM} considers Problem \ref{prob:IntEq} in its entirety and establish existence and uniqueness of is solution within the class of CMIM functions. We also provide an explicit construction for PFPP with inverse marginals that are completely monotonic, see Theorem \ref{thm:PFPP_CMIM} and Algorithm \ref{alg:PFPP_CMIM}. In Section \ref{sec:Examples}, we apply our results to two special cases, namely, the binomial model and the Black-Scholes model. Longer proofs are included in the appendices as well as an excerpt from the theory of the Fourier analysis for tempered distributions.

%
%

\subsection{Notations}\label{sec:notations}
For ease of reference, this subsection provides our frequently used notations.
$\Nb:=\{1,2,\dots\}$ is the set of natural numbers and $\Nb_0:=\{0,1,\dots\}$ is the set of non-negative integers. $\Rb$ is the set of real numbers, and $\Rb_+:=(0,+\infty)$ is the set of positive real numbers. For $t\in\Rb$, $\floor{t}$ is the largest integer that is not larger than $t$ and $\ceil{t}$ is the smallest integer that is not smaller than $t$. For vectors $\thetav_1,\dots,\thetav_n$, we define $(\thetav_1,\dots,\thetav_{n-1})\oplus\thetav_n := (\thetav_1,\dots,\thetav_n)$.

For $\Xc\subseteq \Rb^n$, $\Bc(\Xc)$ denotes the $\sig$-algebra of all the Borel subset of $\Xc$. The support of an $\Rb^n$-valued random variable $\Xv$ (respectively, a measure $\mu$ on $\Rb^n$) is denoted by $\supp(\Xv)$ (respectively, $\supp(\mu)$).

For an open set $D\subseteq\Rb$, $\Cc^n(D)$ denotes the space of all continuously $n$-times differentiable real-valued functions with domain $D$, while $\Cc^{\infty}(D)$ denotes the set of all complex-valued infinitely-differentiable functions with domain $D$. $L^1_\text{loc}$ denotes the set of real-valued functions on $\Rb$ that are integrable on compact subsets of $\Rb$.


%
%
\section{Market setting}\label{sec:Market}

The market consists of a riskless asset and $K\ge 1$ risky assets. We take the riskless asset as the numeraire and denote the discounted prices of the risky asset by the stochastic process $\big(\Sv_t=(S_{t,1},\dots,S_{t,K})\big)_{t\ge0}$.
We will later (see Assumption \ref{asum:rho} below) assume that there exist random variables $\Thetav_n$, $n\in\Nb$, with $\Xi_n:=\supp(\Thetav_n)\subseteq\Rb^{M_n}$ for some $M_n\ge1$. We think of $\Thetav_n$ as the vector of all model parameters for the time period $[n-1,n]$ which are to be learned at the beginning of the period, i.e. at time $n-1$. 
Here, we have fixed a filtered probability space $\big(\Om, \Fc, \Pb, \Fb=(\Fc_t)_{t\ge0}\big)$, in which $\Fb$ is the minimal filtration satisfying the usual conditions such that, for all $t\ge0$, $(\Sv_s)_{s\in[0,t]}$ and $\{\Thetav_n\}_{n=1}^{\ceil{t}}$ are $\Fc_t$-measurable. Henceforth, unless stated otherwise, all stochastic processes are assumed to be $\Fb$-adapted, a.s. stands for $\Pb$-almost surely, and we refer to $(\Fb,\Pb)$-martingales simply as martingales.

The following assumption holds throughout the paper.


\begin{assumption}\label{asum:SPD}
	There exists a unique positive martingale $Z=\big(Z_t\big)_{t\ge0}$ such that: 1) $\Eb[Z_t]=1$ for all $t\ge0$; and 2) $\big(Z_t S_{t,k}\big)_{t\ge0}$ is a martingale for all $k\in\{1,\dots,K\}$. We refer to $Z$ as the \emph{state price density} process.\qed
\end{assumption}

\begin{remark}
	Note that we have \emph{not} assumed the filtration $\Fb$ to be generated by the price process $(\Sv_t)_{t\ge0}$. Thus, the uniqueness of \emph{state price density} does not imply that our market model is complete. For instance, in Examples \ref{ex:Binom} and \ref{ex:BS} below, one cannot replicate, say, the call option $(S_2-K)_+$ over a the time period $[0,2]$. Such a payoff, however, can be replicated over the time period $[1,2)$. In particular, Assumption \ref{asum:SPD} implies that the market is complete over the calibration period $[n-1,n)$ for each $n\in\Nb$, since $(\Fc_s)_{s\in(n-1,n)}$ is only driven by the asset prices.
	
	In light of the above discussion, it is incorrect to refer to our market model as complete. However, the term complete market has already been used in the existing work \cite{AZZ20,LiangStrubWang2021,StrubZhou2021}), with the understanding that completeness means that the market model is complete assuming full knowledge of the model parameters. For this reason, we decided to also refer to our model as a complete model, and reserve the term incomplete to the case in which the model is incomplete even with full knowledge of the model parameters.   \qed
\end{remark}

Our model includes finite-horizon models, despite it being formulated in an infinite-horizon setting.
For $T>0$, let $\big(\Sv_t=(S_{t,1},\dots,S_{t,K})\big)_{t\in[0,T]}$ be the discounted prices in a filtered probability space $\big(\Om, \Fc, \Pb, \Fb=(\Fc_t)_{t\in[0,T]}\big)$, in which $\Fb$ satisfies the usual condition. Assume further that there exists a unique $\Fc_T$-measurable positive random variable $Z_T$ satisfying 1) $\Eb[Z_T]=1$, and 2) $(Z_t\Sv_t)_{t\in[0,T]}$ is a martingale in which $Z_t:=\Eb[Z_T|\Fc_t]$. This model is a special case of our market setting in which $\Fc_t:=\Fc_T$, $\Sv_t:=\Sv_T$, and $Z_t:=Z_T$ for $t>T$.

Our model also embeds discrete-time models, despite being formulated in a continuous-time  setting.
Let $0=t_0<t_1<\dots$ be a given sequence (of times) and $\big\{\Sv_{t_n}=(S_{t_n,1},\dots,S_{t_n,K})\big\}_{n=0}^{+\infty}$ be the discounted prices in a filtered probability space $\big(\Om, \Fc, \Pb, \Fb=\{\Fc_{t_n}\}_{n=0}^{+\infty}\big)$. Assume further that there exists a unique $\Fb$-adapted positive martingale $\{Z_{t_n}\}_{n=0}^{+\infty}$ satisfying 1) $\Eb[Z_{t_n}]=1$, and 2) $\{Z_{t_n}\Sv_{t_n})_{t\in[0,T]}$ is a martingale. This model is embedded in our market setting by defining $\Fc_t:=\Fc_{t_n}$, $\Sv_t:=\Sv_{t_n}$, and $Z_t:=Z_{t_n}$ for $t_n\le t<t_{n+1}$.

To illustrate our results and applicability of our assumptions, we use two benchmark examples. The first Example is the Binomial model proposed by \cite{AZZ20} (who assumed one trading step in each performance evaluation period) and later generalized by \cite{LiangStrubWang2021} (who assumed multiple trading steps in each performance evaluation period).

\begin{example}[The generalized Binomial model of \cite{AZZ20} and \cite{LiangStrubWang2021}]\label{ex:Binom}
Consider a discrete-time model with $K=1$ and $t_n:= n/N$ for some constant $N\in\Nb$.
Let $\{u_n\}_{n\in\Nb}$, $\{d_n\}_{n\in\Nb}$, $\{p_n\}_{n\in\Nb}$, and $\{B_n\}_{n\in\Nb}$ be sequences of random variables in a probability space $(\Omega,\Fc,\Pb)$ such that, for all $n\in\Nb$, $B_n\in\{0,1\}$ (i.e. it is a Bernoulli random variable), $d_n,p_n\in(0,1)$, and $u_n>1$ a.s..

Define the filtration $\Fb=\{\Fc_{n/N}\}_{n\in\Nb_0}$ such that $\Fc_{n/N}$ is the augmented $\sigma$-field generated by $\{B_i\}_{i=1}^n$ and $\{(u_j,d_j,p_j)\}_{j=1}^{N\ceil{(n+1)/N}}$. That is, $\Fc_0$ is generated by $\{(u_j,d_j,p_j)\}_{j=1}^{N}$; $\Fc_{1/N}$ is generated by $B_1$ and $\{(u_j,d_j,p_j)\}_{j=1}^{N}$; ... ; $\Fc_{N-1/N}$ is generated by $\{B_n\}_{n=1}^{N-1}$ and $\{(u_j,d_j,p_j)\}_{j=1}^{N}$; $\Fc_1$ is generated by $\{B_n\}_{n=1}^N$ and $\{(u_j,d_j,p_j)\}_{j=1}^{2N}$; and so forth. Note that $B_n$ is $\Fc_{n/N}$-measurable, while $(u_j,d_j,p_j)$ are $\Fc_{\floor{j/N}}$-measurable.\footnote{$\floor{t}$ is the floor function, that is, the largest integer that is not larger than $t$.} In other words, $B_n$, $n\in\Nb$, is revealed ``one-at-a-time'' and at time $n/N$, the end of the $n$-th ``trading period'' $[\frac{n-1}{N},\frac{n}{N})$. The binomial parameters $\{(u_j,d_j,p_j)\}_{j=kN+1}^{(k+1)N}$ are revealed ``N-at-a-time'' and at time $k\in\Nb_0$, the start of the $k$-th ``evaluation period'' $[k,k+1)$. 

Note that the parameters of the binomial model $\{(u_j,d_j,p_j)\}_{j\in\Nb}$ are random variables and change through time, and that the filtration $\Fb$ is such that the binomial parameters for the time period $[k,k+1]$ are know at time $k\in\Nb_0$. In other words, we have assumed that the model is calibrated at each time $k\in\Nb_0$ and the estimated parameters $\{(u_j,d_j,p_j)\}_{j=kN+1}^{(k+1)N}$ are believed to be correct during the time period $[k,k+1]$.

Assume that the prices $\{S_{n/N}\}_{n\in\Nb_0}$ are given recursively by $S_{n/N} = S_{(n-1)/N}\big(u_n B_n + d_n (1-B_n)\big)$, with $S_0=1$. 	
Assume further that $\Eb[B_n|\Fc_{n-1}] = p_n$ a.s., which implies that $p_n$ is the conditional probability of an upward jump during the $n$-th trading period since
\begin{align}\label{eq:Binom_p}
\Pb\left(S_{\frac{n}{N}}>S_{\frac{n-1}{N}}\middle|\Fc_{n-1}\right) = \Pb(B_n=1|\Fc_{n-1})=\Eb[B_n|\Fc_{n-1}] = p_n;\quad n\in\Nb.
\end{align}	
For this model, the state-price-density process $\{Z_{n/N}\}_{n\in\Nb_0}$ is given by
\begin{align}
Z_{\frac{n}{N}} = Z_{\frac{n-1}{N}}\left(\frac{q_n}{p_n} B_n + \frac{1-q_n}{1-p_n}(1-B_n)\right);\quad n\in\Nb,
\end{align}
with $Z_0=1$, in which $q_n:=(1-d_n)/(u_n-d_n)$. As we argued before, this discrete-time model is a special case of our model by setting $\Fc_t:=\Fc_{t_n}$, $\Sv_t:=\Sv_{t_n}$, and $Z_t:=Z_{t_n}$ for $t_n\le t<t_{n+1}$. \qed
\end{example}\vspace{1ex}

Our second benchmark example is an extension of the Black-Scholes model in which the parameters (i.e. the drift and diffusion coefficients) are random variables that are learned through time. We assume that the market model is calibrated at discrete times $n\in\Nb_0$ and the (estimated) parameters are believed to be correct for the time period $[n,n+1]$. This model is a continuous-time model and, to the best of our knowledge, existence of PFPPs has not been established in any continuous-time model.\footnote{Here, we are referring to existence of a PFPP in a multi-period setting. For the single-period setting, that is, finite-horizon problem in which the value function is given and the terminal utility function is unknown, there are existing results such as \cite{Kallblad2020}  for the Black-Scholes model and \cite{StrubZhou2021} for a complete semi-martingale setting.} 

\begin{example}[A generalized Black-Scholes market]\label{ex:BS}
Let $\Bv=(\Bv_t)_{t\ge0}$ be a $K$-dimensional standard Brownian motion (as before, $K\in\Nb$ is the number of risky assets), $\{\Lamv_n\}_{n\in\Nb}$ be a sequence of $\Rb^K$-valued random variables, and $\Sig=\{\Sig_n\}_{n\in\Nb}$ be a sequence of $K\times K$ non-singular random matrices in a filtered probability space $\big(\Om, \Fc, \Pb, \Fb=(\Fc_t)_{t\ge0}\big)$. We assume that $(\Bv_t-\Bv_n)_{t\ge n}$ is independent of $\{(\Lamv_m,\Sig_m)\}_{m=1}^{n+1}$ for all $n\in\Nb_0$, and that $\Fc_t$ is the augmented $\sig$-field generated by $(\Bv_s)_{0\le s\le t}$, $\{\Lamv_n\}_{1\le n\le \ceil{t}}$, and $\{\Sig_n\}_{1\le n\le \ceil{t}}$. In particular, $\Bv_t$, $t\ge0$, is $\Fc_t$-measurable while $(\Lamv_n,\Sig_n)$, $n\in\Nb$, is $\Fc_{n-1}$-measurable.
Let $\Sv=\big(\Sv_t=(S_{t,1},\dots,S_{t,K})\big)_{t\ge0}$ be the strong solution of
\begin{align}\label{eq:S_BS}
\dd \Sv_t = \diag(\Sv_t)\Sig_n(\Lamv_n \dd t + \dd \Bv_t);\quad n-1\le t<n, n\in\Nb,
\end{align}
with $\Sv_0=\sv>0$. For this model, the state price density process is given by
\begin{align}
Z_t = Z_{n-1} \exp\left(-\frac{1}{2}\|\Lamv_n\|^2(t-n+1) -\Lamv_n^\top(\Bv_t-\Bv_{n-1})\right);
\quad n-1\le t<n, n\in\Nb,
\end{align}
and with $Z_0=1$.
\qed
\end{example}\vspace{1ex}

Next, we introduce the set of admissible wealth processes. Throughout the paper, we abstract away the investment policy (i.e. the portfolio weights of the risky assets through time), as it is implied by the standard replication argument for complete markets. 
\begin{definition}\label{def:admiss}
A process $X=(X_t)_{t\ge0}$ is a \emph{wealth processes} if $(Z_t X_t)_{t\ge0}$ is a martingale and $X_t\ge0$ a.s. for all $t\ge0$. We denote the set of all admissible wealth process by $\Act$.\qed
\end{definition}

As it will be clear from Definition \ref{def:PFPP} below, discrete forward performance processes rely only on the observed values of wealth processes at discrete times $t\in\Nb_0$, rather than the whole path of wealth processes. This motivates the definition of \emph{discretely-observed} wealth processes.
Let us introduce the set
\begin{align}\label{eq:Act}
\Ac := \Big\{\{X_n\}_{n\in\Nb_0}: X_0\text{ is $\Fc_0$-measurable},~ X_0\ge0~\text{a.s.},~ X_n\in\Ac_n(X_{n-1})\text{ for } n\in\Nb\Big\},
\end{align}
in which we have defined the sets
\begin{align}\label{eq:Ac_n}
\Ac_n(\xi) := \Big\{X: X\text{ is $\Fc_n$-measurable},~ X\ge0~\text{a.s.},~ \Eb\left[X Z_n\middle|\Fc_{n-1}\right]=\xi Z_{n-1}\Big\},
\end{align}
for any $n\in\Nb$ and any $\Fc_{n-1}$-measurable non-negative random variable $\xi$. We interpret $\Ac$ as the set of admissible wealth processes observed at discrete times $n\in\Nb_0$.

As the following lemma shows, there is a one-to-one correspondence between the set of continuously observed admissible wealth processes $\Act$ and the set of discretely observed admissible wealth processes $\Ac$. Indeed, since the market is complete, we can recover $(X_t)_{t\in[n-1,n)}$ from $X_n$ by the relationship $X_t = \Eb[X_n Z_n/Z_t|\Fc_t]$. In light of this fact, we do not distinguish between a wealth process $(X_t)_{t\ge0}$ and its discretely observed counterpart $\{X_n\}_{n\in\Nb_0}$.
\begin{lemma}\label{lem:DisceretWealth}
If $(X_t)_{t\ge0}\in\Act$, then $\{X_n\}_{n\in\Nb_0}\in\Ac$. Conversely, let $\{X_n\}_{n\in\Nb_0}\in\Ac$ and define $X_t:=\Eb\left[X_{\ceil{t}}\frac{Z_{\ceil{t}}}{Z_t}\middle|\Fc_t\right]$.\footnote{$\ceil{t}$ is the ceiling function (i.e. the smallest integer that is not smaller than $t$).} Then, $\left(X_t\right)_{t\ge0}\in\Act$.\qed
\end{lemma}
\begin{proof}
The first statement directly follows from the fact that, if $X\in\Act$, then $\{X_nZ_n\}_{n=0}^{\infty}$ is a non-negative martingale and, therefore, $X_n\in\Ac_n(X_{n-1})$ for all $n\in\Nb$. To show the converse statement, let $\{X_n\}_{n=0}^{+\infty}\in\Ac$ and define $X=(X_t)_{t\ge0}$ as in the statement of the lemma. For $t\ge0$, we have that $X_t=\Eb\left[X_{\ceil{t}}Z_{\ceil{t}}/Z_t|\Fc_t\right]\ge0$ a.s. since $X_{\ceil{t}}\ge 0$ a.s.  by \eqref{eq:Act}. It only remains to show that $(X_tZ_t)_{t\ge0}$ is a martingale. If $n-1\le s<t\le n$ for some $n\in\Nb$, then $\Eb[X_tZ_t|\Fc_s]=\Eb\big[\Eb[X_nZ_n|\Fc_t]\big|\Fc_s\big]=X_sZ_s$. If $n-1\le s< n < t\le n+1$ for some $n\in\Nb$, then
$\Eb[X_tZ_t|\Fc_s]=\Eb\big[\Eb[X_tZ_t|\Fc_n]\big|\Fc_s\big]=\Eb\big[X_nZ_n\big|\Fc_s\big]=X_sZ_s$. Using induction, it then follows that $\Eb[X_t|\Fc_s]=X_s$ for all $t>s\ge0$. So, $X\in\Act$.
\end{proof}

The following assumption plays a central rule in our definition of PFPPs as well as the arguments and proofs in later sections. In short, it requires the existence of a sequence of random vectors $\{\Thetav_n\}_{n\in\Nb}$ such that the $\Fc_n$-measurable random variable $\rho_n:=Z_n/Z_{n-1}$ (i.e. the so-called \emph{pricing kernel} for time period $[n-1,n]$) is conditionally independent of $\Fc_{n-1}$ given $(\Thetav_1,\dots,\Thetav_n)$. As mentioned at the beginning of this section, we think of $\Thetav_n$ as the vector of all model parameters for the time period $[n-1,n]$. Therefore, it is reasonable to assume that $\Thetav_n$ is $\Fc_{n-1}$-measurable.

\begin{assumption}\label{asum:rho}
Let $(Z_t)_{t\ge0}$ be as in Assumption \ref{asum:SPD} and define $\rho_n :=Z_n/Z_{n-1}$ for $n\in\Nb$. There exist $\Fc_{n-1}$-measurable random variables $\Thetav_n$, $n\in\Nb$, with $\Xi_n:=\supp(\Thetav_n)\subseteq\Rb^{M_n}$ for some $M_n\ge1$, such that
\begin{align}\label{eq:CondInd}
\Pb(\rho_n\in B|\Fc_{n-1}):=\Eb\left[\Ib_{\{\rho_n\in B\}}\middle| \Fc_{n-1}\right]=
\Eb\left[\Ib_{\{\rho_n\in B\}}\middle| \Thetav_1,\dots,\Thetav_n\right],\quad \as, n\in\Nb, B\in\Bc(\Rb_+).
\end{align}
For ease of notation, we define $\Gv_n := (\Thetav_1,\dots,\Thetav_n)$, $n\in\Nb$, and denote $\Gc_n:=\supp(\Gv_n)\subseteq \Xi_1\times\dots\times \Xi_n$. With a slight abuse of notation, we take the convention that $\Gv_1=\Thetav_1$ and $\Gc_1=\Xi_1$. Thus, \eqref{eq:CondInd} becomes $\Pb(\rho_n\in B|\Fc_{n-1})=\Pb(\rho_n\in B|\Gv_n)$, for all $(n,B)\in\Nb\times\Bc(\Rb_+)$. 
\qed
\end{assumption}

As the following remark indicates, Assumption \ref{asum:rho} is satisfied in our two benchmark models, namely, the generalized binomial model of Example \ref{ex:Binom} and the generalized Black-Scholes market of Example \ref{ex:BS}. Note that in both cases, $\Thetav_n$ is the vector of all model parameters for the time period $[n-1,n]$.

\begin{remark}\label{rem:Binom_BS_rho}
In Example \ref{ex:Binom}, we have that
\begin{align}\label{eq:Rho_binom}
\rho_n := \frac{Z_n}{Z_{n-1}}= \prod_{m=1+(n-1)N}^{nN}\left(\frac{q_m}{p_m} B_m + \frac{1-q_m}{1-p_m}(1-B_m)\right);\quad n\in\Nb.
\end{align}
For $n\in\Nb$, let $\Thetav_n=\{(u_m,d_m,p_m)\}_{m=1+(n-1)N}^{nN}$ and note that $\Thetav_n$ is $\Fc_{n-1}$-measurable. Then, \eqref{eq:Binom_p} yields that
$\Pb(B_m=1|\Fc_{n-1})=\Eb[B_m|\Fc_{n-1}] = p_m
= \Eb[p_m|\Thetav_n]=\Eb\big[\Eb[B_m|\Fc_{n-1}]\big|\Thetav_n\big] = \Eb[B_m |\Thetav_n]=\Pb(B_m=1|\Thetav_n)$, 
for $n\in\Nb$ and $m\in\{1+(n-1)N,\dots, nN\}$.
Therefore, Assumption \ref{asum:rho} holds since $\Eb\left[\Ib_{\{\rho_n\le t\}}\middle| \Fc_{n-1}\right]=\Eb\left[\Ib_{\{\rho_n\le t\}}\middle| \Thetav_n\right]$ for all $n\in\Nb$ and $t\in\Rb$. 	

In Example \ref{ex:BS}, $\rho_n :=Z_{n}/Z_{n-1} = \exp\left(-\frac{1}{2}\|\Lamv_n\|^2 - \Lamv_n^\top(\Bv_n-\Bv_{n-1})\right)$. Since $\Lamv_n$ is $\Fc_{n-1}$-measurable and $(\Bv_n-\Bv_{n-1})$ is independent of $\Fc_{n-1}$, we have that $\Eb\left[\Ib_{\{\rho_n\le t\}}\middle| \Fc_{n-1}\right]=
\Eb\left[\Ib_{\{\rho_n\le t\}}\middle| \Lamv_n\right]$ for $n\in\Nb$ and $t\in\Rb$. Thus, Assumption $\ref{asum:rho}$ holds for $\Thetav_n := \Lamv_n$.
\qed
\end{remark}

One of the main contribution of our paper is to highlight the role of Assumption \ref{asum:rho} in establishing existence conditions and providing a construction algorithm for PFPPs. As we will discuss in the next section, a PFPP is a sequence of random utility functions $(x,\om)\mapsto U_n(x,\om)$, $(n,x,\om)\in\Nb_0\times\Rb^+\times\Om$. More specifically, the random function $U_n(\cdot):=U_n(\cdot,\om)$ is measurable with respect to a sub-$\sig$-algebra $\mathcal{G}_{n-1}\subseteq\Fc_{n-1}$ (which is why these preferences are \emph{predictable}). Assumption \ref{asum:rho} allows us to express the subfiltration $\{\mathcal{G}_n\}_{n\in\Nb_0}$ more explicitly than what was used for PFPPs in the existing work \cite{AZZ20} and \cite{StrubZhou2021}. In particular, by adapting Assumption \ref{asum:rho} and taking $\mathcal{G}_{n-1}$ to be the augmented $\sig$-algebra generated by $\Gv_n := (\Thetav_1,\dots,\Thetav_n)$, we are able to define a PFPP as a sequence $\{U_n(\cdot,\Gv_n)\}_{n\in\Rb_0}$ in which $x,\gv\mapsto U_n(x,\gv)$ is a \emph{deterministic} measurable function, see Definition \ref{def:PFPP} below. The advantage of working with measurable functions (instead of random fields) is that it leads to more explicit existence conditions for PFPPs (see Theorem \ref{thm:verif} below) that are reduced to a single period integral equation (see subsection \ref{sec:PFPP_construction}). Such existence conditions and, more importantly, a rigorous argument establishing how they are related to a single period problem, have been missing in the literature beyond the existence result of \cite{AZZ20} for the binomial model.
In a more abstract setup such as \cite{StrubZhou2021}, we speculate that one should also assume a counterpart of Assumption \ref{asum:rho} to obtain an existence results for PFPPs. However, to keep the argument less technical, we have refrain from using a more general setting and will consider such an extension as future work.

%
%
\section{Predictable forward performance processes}\label{sec:PFPP}

In this section, we define PFPPs and provide conditions for their existence, see Theorem \ref{thm:verif}. Based on the existence conditions, we then propose a forward (in time) period-by-period construction of PFPPs in subsection \ref{sec:PFPP_construction}. In each period, the main step of the construction is solving an integral equation (namely, \eqref{eq:IntEq_Stepn} below) which will be analyzed in Section \ref{sec:ConvEq}.

Motivated by \cite{AZZ20}, we define predictable forward performance processes as follows. Recall from Assumption \ref{asum:rho} that $\rho_n:=Z_n/Z_{n-1}$, that $\{\Thetav_n\}_{n\in\Nb}$ is a sequence of $\{\Fc_n\}_{n\in\Nb_0}$-predictable random vectors (i.e. $\Thetav_n$ is $\Fc_{n-1}$-measurable) satisfying \eqref{eq:CondInd}, and that $\Gv_n:=(\Thetav_1,\dots,\Thetav_n)$. Furthermore, let $\Uc$ be the set of classical utility function on $\Rb_+$, namely,
\begin{align}
\Uc := \left\{U\in\Cc^2(\Rb_+): U'>0, U''<0, U'(0+)=+\infty, U'(+\infty)=0\right\},
\end{align}
and let $\oplus$ be the direct sum of vectors such that $(\Thetav_1,\dots,\Thetav_{n-1})\oplus\Thetav_n := (\Thetav_1,\dots,\Thetav_n)$.

\begin{definition}\label{def:PFPP}
Consider the market setting of Section \ref{sec:Market} with $\{(\rho_n,\Thetav_n,\Gv_n)\}_{n\in\Nb}$ as in Assumption \ref{asum:rho} and recall that $\Xi_n:=\supp(\Thetav_n)$ and $\Gc_n:=\supp(\Gv_n)$.
A sequence $\{U_n\}_{n\in\Nb_0}$ of Borel measurable functions $U_0:\Rb_+\to\Rb$ and $U_n:\Rb_+\times \Gc_n \to \Rb$, $n\in\Nb$,
is a \emph{predictable forward performance process (PFPP)} if the following conditions are satisfied:
\begin{enumerate}
\item[$(i)$] $U_0\in\Uc$ and $U_n(\cdot,\gv)\in\Uc$ for all $(n,\gv)\in\Nb\times\Gc_n$.

\item[$(ii)$] $U_{n-1}(x,\gv')\ge \Eb\big[U_n(X,\gv)\big|\Gv_n=\gv\big]$ for all $(n,x,\gv=\gv'\oplus\thetav)\in\Nb\times\Rb_+\times\Gc_n$ (such that $\gv'\in\Gc_{n-1}$ and $\thetav\in\Xi_n$) and for any $X\in\Ac_n(x)$ satisfying $\Eb\big[U_n(X,\gv)\big|\Gv_n=\gv\big]>-\infty$.\footnote{For $n=1$, this condition becomes $U_0(x)\ge\Eb\big[U_1(X,\thetav)\big|\Thetav_1=\thetav\big]$ for all $(x,\thetav)\in\Rb_+\times \Xi_1$ and $X\in\Ac_1(x)$ such that $\Eb\big[U_1(X,\thetav)\big|\Thetav_1=\thetav\big]>-\infty$.}

\item[$(iii)$] There exists $(X^*_t)_{t\ge0}\in\Act$ such that $U_{n-1}(x,\gv') = \Eb\big[U_n(X^*_n,\gv)\big|X^*_{n-1}=x,\Gv_n=\gv\big]$ for all $(n,x,\gv=\gv'\oplus\thetav)\in\Nb\times\Rb_+\times\Gc_n$ (such that $\gv'\in\Gc_{n-1}$ and $\thetav\in\Xi_n$).\footnote{For $n=1$, this condition becomes $U_0(x)=\Eb\big[U_1(X^*_1,\thetav)\big|X^*_0=x,\Thetav_1=\thetav\big]$ for all $(x,\thetav)\in\Rb_+\times\Xi_1$.}
\end{enumerate}
The wealth process $(X^*_t)_{t\ge0}$ in (iii) is called an optimal wealth process for PFPP $\{U_n\}_{n\in\Nb}$.
\qed
\end{definition}\vspace{1ex}

\begin{remark}
The condition $\Eb\big[U_1(X,\thetav)\big|\Thetav_1=\thetav\big]>-\infty$ is included in Definition 3.1.(ii) since strategies for which $\Eb\big[U_n(X,\gv)\big|\Gv_n=\gv\big]=-\infty$ are clearly sub optimal and do not need to be checked.\qed
\end{remark}

One can think of a PFPP as a sequence of utility functions for an agent such that the agent's preference at time $n$ is quantified by $U_n(X_n,\Gv_n)$. Condition $(ii)$ of Definition \ref{def:PFPP} states that, for an arbitrary wealth process $(X_t)_{t\ge0}\in\Act$, the stochastic process $\{U_n(X_n,\Gv_n)\}_{n\in\Nb_0}$ is a super martingale. For an optimal $(X^*_t)_{t\ge0}\in\Act$, Condition $(iii)$ implies that $\{U_n(X^*_n,\Gv_n)\}_{n\in\Nb_0}$ is a martingale. Thus, Properties $(ii)$ and $(iii)$ are Bellman's \emph{dynamic programming principles} and enforce \emph{time-consistency} for PFPPs. See \cite{AZZ20} for a more detailed discussion.

Note, also, that our definition of PFPPs is more restricted than the one in Definition 2.1 of \cite{AZZ20} and \cite{StrubZhou2021}. In those studies, a PFPP is a sequence of \emph{random} function $(x,\om)\to \Ut_n(x,\om)$, $(n,x,\om)\in\Nb_0\times\Rb_+\times\Om$, such that $\Ut_n(x,\om)$ is $\Fc_{n-1}$-measurable. We have defined a PFPP as a sequence of \emph{deterministic} measurable functions $(x,\gv)\mapsto U_n(x,\gv)$, $n\in\Nb_0$. By defining $\Ut(x,\om):=U_n\Big(x,\big(\Thetav_1(\om),\dots,\Thetav_n(\om)\big)\Big)$, one can check that $\{\Ut_n\}_{n\in\Nb_0}$ is a PFPP according to \cite{AZZ20}, but, with a more restrictive measurability condition. In particular, Definition \ref{def:PFPP} implies that $\Ut_n(x,\om)$ is measurable with respect to the (augmented) $\sig$-algebra generated by $(\Thetav_1,\dots,\Thetav_n)$, which is a sub-$\sig$-algebra of $\Fc_{n-1}$ since $\Thetav_n$ is $\Fc_{n-1}$-measurable by Assumption \ref{asum:rho}.
This more restricted definition (along with Assumption \ref{asum:rho}) allows us to: 1) find existence conditions for PFPPs in multi-period settings, and 2) show that our multi-period existence conditions reduce to a single period integral equation. The next two subsections elaborate these two results.

\subsection{Existence of PFPPs}\label{sec:Verify}

Our first goal in this section is to provide a set of conditions for existence of PFPPs. These conditions and their proof  rely on \emph{inverse marginal} and \emph{convex dual} functions of the utility functions $U_n(\cdot,\gv)$, $(n,\gv)\in\Nb_0\times\Gc_n$. The following lemma provide the basic properties of these well-known functions. In its statement,
\begin{align}\label{eq:Ic}
\Ic := \left\{I\in\Cc^1(\Rb_+): I'<0, I(0+)=+\infty, I(+\infty)=0\right\},
\end{align}
denotes the set of inverse marginal functions.

\begin{lemma}\label{lem:ConvexDuality}
A utility function $U(\cdot)\in\Uc$ has a unique inverse marginal function $I\in\Ic$ defined by $U'\big(I(y)\big)=y$, $y>0$, and a unique \emph{convex dual function} $V:\Rb_+\to\Rb$ given by
\begin{align}\label{eq:ConvDual}
V(y) := \sup_{x>0}\left\{U(x) - xy\right\} = U\big(I(y)\big)-yI(y);\quad y>0,
\end{align}
which is in $\Cc^2(\Rb_+)$, strictly decreasing, and strictly convex. Furthermore, we have $V'(y)=-I(y)$ and $V''(y)=-I'(y)=-1/U''\big(I(y)\big)$, $y\ge0$.\qed
\end{lemma}
\begin{proof}
The proof is simple and can be found in many standard texts on convex analysis. See, for instance, Theorem 26.5 of \cite{Rockafellar1970}.
\end{proof}  

The following theorem is our first main result of the paper. It provides a set of sufficient conditions for a sequence of functions to be a PFPP. To the best of our knowledge, there has not been such a result in the literature beyond the existence result of \cite{AZZ20} and \cite{LiangStrubWang2021} for the binomial model.\footnote{We emphasize again that by an existence result for PFPPs, we mean conditions for existence of a PFPP in a multi-period evaluation setting. In particular, we do not claim that we are the first to provide existence of a solution for the inverse Merton problem, which can be seen as a special case of Theorem \ref{thm:verif} in a finite-horizon model with only one evaluation period. For the inverse Merton problem, there are existing results such as \cite{Kallblad2020}  for the Black-Scholes model and \cite{StrubZhou2021} for a complete semi-martingale setting.} 

\begin{theorem}\label{thm:verif}
Consider the market setting of Section \ref{sec:Market} with Assumptions \ref{asum:SPD} and \ref{asum:rho} holding. Let $U_0\in \Uc$ and $I_0:=U_0^{\prime-1}\in\Ic$. Furthermore, assume that Borel measurable functions $I_n:\Rb_+\times\Gc_n\to \Rb_+$, $n\in\Nb$, satisfy the following conditions for all $(n,y,\gv)\in\Nb\times\Rb_+\times\Gc_n$:

\noindent$(i)$ $I_n(\cdot,\gv)\in\Ic$ and $\Eb\left[I_n(y\rho_n, \gv)\middle|\Gv_n=\gv\right]<+\infty$.

\noindent$(ii)$ $\Eb\left[\rho_n I_n(y\rho_n, \gv)\middle|\Gv_n=\gv\right]=I_{n-1}(y,\gv')$, in which $\gv=\gv'\oplus\thetav$ with $\gv'\in\Gc_{n-1}$ and $\thetav\in\Xi_n$.\footnote{For $n=1$, this condition becomes $\Eb\left[\rho_1 I_1(y\rho_1, \thetav)\middle|\Thetav_1=\thetav\right]=I_0(y)$ for all $(y,\thetav)\in\Rb_+\times\Xi_1$.}\vspace{1ex}

\noindent For $n\in\Nb$, define $U_n:\Rb_+\times\Gc_n\to\Rb$ by
\begin{align}\label{eq:U}
U_n(x,\gv) &:= U_{n-1}\left(I_{n-1}(1,\gv'),\gv'\right) +\Eb\left[\int_{I_n\left(\rho_n, \gv\right)}^x I_n^{-1}(\xi,\gv) \dd\xi\middle|\Gv_n=\gv\right],
\end{align}
for $x\in\Rb_+$ and $\gv=\gv'\oplus\thetav\in\Gc_n$ (such that $\gv'\in\Gc_{n-1}$ and $\thetav\in\Xi_n$).\footnote{For $n=1$, \eqref{eq:U} becomes $U_1(x,\thetav) := U_0\big(I_0(1)\big) +\Eb\left[\int_{I_1\left(\rho_1,\thetav\right)}^x I_1^{-1}(\xi,\thetav) \dd\xi\middle|\Thetav_1=\thetav\right]$, $(x,\thetav)\in \Rb_+\times \Xi_1$.} For an $x_0>0$, let $X^*_0=x_0$ and
\begin{align}\label{eq:XStar}
X^*_n:=I_n\left(\rho_n I_{n-1}^{-1}(X^*_{n-1},\Gv_{n-1}), \Gv_n\right);\quad n\in\Nb.
\end{align}
Then, $\left\{U_n\right\}_{n\in\Nb_0}$ is a PFPP and $\{X^*_n\}_{n\in\Nb_0}$ is a corresponding optimal wealth process.\qed
\end{theorem}

\begin{proof}
See Appendix \ref{app:verif}.
\end{proof}

Before going further, let us highlight the role of Assumption \ref{asum:rho}. It may seem at first that this assumption only plays a minor role in the proof of Theorem \ref{thm:verif} in that it is only needed to obtain \eqref{eq:AssumpRho_needed}. In fact, one may argue that the proof can be generalized by replacing $U_n(\cdot,\Gv_n)$ with a more general $\Fc_{n-1}$-measurable random field $U_n(\cdot,\om)$ as in Definition 2.1 of \cite{AZZ20} and \cite{StrubZhou2021}. We agree that such a generalization of Theorem \ref{thm:verif} is possible.

The difficulty, however, is in how the resulting existence conditions can be used for constructing PFPPs. In particular, how such more abstract conditions could be rigorously reduced to a single period problem (in our case, the integral equation \eqref{eq:IntEq_Stepn} below). Because of this issue, \cite{AZZ20} only provided existence conditions for PFPPs in the binomial setting. Furthermore, their construction algorithm for PFPPs (see Theorem 7.1 on page 340 of \cite{AZZ20}) only produces PFPPs that are of the form $U_n(x,\Gv_n)$, $n\in\Nb_0$, in which $\Gv_n=(\Thetav_1,\dots,\Thetav_n)$ are as in Remark \ref{rem:Binom_BS_rho} (for the binomial setting of Example \ref{ex:Binom}). In short, although the Definition of PFPPs in \cite{AZZ20} is more general than ours, their concrete results are special case of ours. \cite{StrubZhou2021} faced a similar difficulty and, as they explicitly mention on two occasions, they did not provide any multi-period existence condition for PFPPs.\footnote{See Section \ref{sec:intro} for details.}

Because of Assumption \ref{asum:rho}, we are able to express randomness of PFPPs through the random variables $\Gv_n=(\Thetav_1,\dots,\Thetav_n)$, as we have done in Definition \ref{def:PFPP}. Furthermore, \eqref{eq:AssumpRho_needed} shows that Assumption \ref{asum:rho} is necessary for such a representation.

\subsection{Forward construction of PFPPs}\label{sec:PFPP_construction}
Our second goal in Section \ref{sec:PFPP} is to find an algorithm for constructing a PFPP $U_n(x,\gv)$, $(n,x,\gv)\in\Nb_0\times\Rb_+\times\Gc_n$, using the existence conditions provided by Theorem \ref{thm:verif}. As elaborated in the introduction, we are interested in a forward-in-time construction. That is, we would like to iteratively obtain $U_n(\cdot,\Gv_n)$ assuming that we know $U_{n-1}(\cdot,\Gv_{n-1})$.

Assume that $U_0\in\Uc$ is given a priori (i.e. at time 0) and let $I_0\in\Ic$ be its inverse marginal. Condition $(ii)$ of Theorem \ref{thm:verif} dictates that
\begin{align}\label{eq:IntEq_Step0}
\Eb\left[\rho_1 I_1(y\rho_1, \thetav)\middle|\Thetav_1=\thetav\right]=I_0(y);&\quad (y,\thetav)\in\Rb_+\times\Xi_1.
\end{align}
Here, $I_0$ and the conditional distribution of $\rho_1|_{\Thetav_1=\thetav}$ are known at time $0$, while $I_1$ is unknown. Furthermore, for Condition $(i)$ of Theorem \ref{thm:verif} to be satisfied, we also require that $I_1(\cdot,\thetav)\in\Ic$ and $\Eb\left[I_1(y\rho_1, \thetav)\middle|\Thetav_1=\thetav\right]<+\infty$ for all $y>0$ and $\thetav\in \Xi_1$. Finding such an $I_1$ is formulated as Problem \ref{prob:PFPPConstructionStep} below, which will be the focus of Section \ref{sec:ConvEq}. Once we find $I_1$, we may use \eqref{eq:U} and $\eqref{eq:XStar}$ to define $X^*_1$ and $U_1$ as follows,
\begin{align}\label{eq:XStarU_Step0}
\begin{cases}
X^*_1 :=I_1\left(\rho_1 I_0^{-1}(x_0), \Thetav_1\right)=I_1\left(\rho_1 U_0'(x_0), \Thetav_1\right),\\
U_1(x,\thetav) := U_0\big(I_0(1)\big) +\Eb\left[\int_{I_1\left(\rho_1,\thetav\right)}^x I_1^{-1}(\xi,\thetav) \dd\xi\middle|\Thetav_1=\thetav\right];\quad x>0, \thetav\in\Xi_1,
\end{cases}
\end{align}
in which $x_0>0$ is the initial portfolio value.

Next, consider the second evaluation period $t\in[1,2)$. At $t=1$, we know $I_1(\cdot,\thetav_1)$ and the conditional distribution of $\rho_2|_{\Gv_2=(\thetav_1,\thetav_2)}$. From Conditions $(i)-(ii)$ of Theorem \ref{thm:verif}, we are looking for an $I_2$ such that $I_2\big(\cdot,(\thetav_1,\thetav_2)\big)\in\Ic$, $\Eb\left[I_2(y\rho_2,(\thetav_1,\thetav_2)\big)\middle|\Gv_2=(\thetav_1,\thetav_2)\right]<+\infty$, and
\begin{align}\label{eq:IntEq_Step1}
\Eb\left[\rho_2 I_2\big(y\rho_2, (\thetav_1,\thetav_2)\big)\middle|\Gv_2=(\thetav_1,\thetav_2)\right] = I_1(y,\thetav_1),
\end{align}
for all $y>0$ and $(\thetav_1,\thetav_2)\in\Gc_2$. Finding such an $I_2$ is also formulated as Problem \ref{prob:PFPPConstructionStep} below, which we solve in the next section. Once an appropriate $I_2$ is found, we then obtain $U_2$ and $X^*_2$ by \eqref{eq:U} and \eqref{eq:XStar} respectively, that is,
\begin{align}\label{eq:XStarU_Step1}
\begin{cases}
X^*_2 :=I_2\big(\rho_2 I_1^{-1}(X^*_1,\Thetav_1), (\Thetav_1,\Thetav_2)\big),\\
U_2\big(x,(\thetav_1,\thetav_2)\big) := U_1\big(I_1(1,\thetav_1),\thetav_1\big)\\
{}+\Eb\left[\int_{I_2\big(\rho_2,(\thetav_1,\thetav_2)\big)}^x I_2^{-1}\big(\xi,(\thetav_1,\thetav_2)\big) \dd\xi\middle|\Gv_2=(\thetav_1,\thetav_2)\right];\quad x>0,(\thetav_1,\thetav_2)\in\Gc_2.
\end{cases}
\end{align}
Note that \eqref{eq:IntEq_Step1} and \eqref{eq:XStarU_Step1}
can be solved at time $1$ (specifically, recall that $\Gv_2:=(\Thetav_1,\Thetav_2)$ is $\Fc_1$ measurable by Assumption \ref{asum:rho}).

In general, at time $n-1\in\Nb_0$, we are given $U_{n-1}$, $I_{n-1}$, $\Gv_n:=(\Thetav_1,\dots,\Thetav_n)$ and the conditional distribution of $\rho_n|_{\Gv_n}$. Using the results of the next section, we first find an $I_n$ satisfying Conditions $(i)$ and $(ii)$ of Theorem \ref{thm:verif} by solving the equation
\begin{align}\label{eq:IntEq_Stepn}
\Eb\left[\rho_n I_n\big(y\rho_n, (\thetav_1,\dots,\thetav_n)\big)\middle|\Gv_n=(\thetav_1,\dots,\thetav_n)\right] = I_{n-1}\big(y,(\thetav_1,\dots,\thetav_{n-1})\big),
\end{align}
for all $y>0$ and $(\thetav_1,\dots,\thetav_n)\in\Gc_n$. Then, we obtain $X^*_n$ and $U_n$ as follows
\begin{align}\label{eq:XStarU_Stepn}
\begin{cases}
X^*_n :=I_n\big(\rho_n I_{n-1}^{-1}(X^*_{n-1},\Gv_{n-1}), \Gv_n)\big),\\
U_n\big(x,(\thetav_1,\dots,\thetav_n)\big) := U_{n-1}\Big(I_{n-1}\big(1,(\thetav_1,\dots,\thetav_{n-1})\big),(\thetav_1,\dots,\thetav_{n-1})\Big)\\
{}+\Eb\left[\int_{I_n\big(\rho_n,(\thetav_1,\dots,\thetav_n)\big)}^x I_n^{-1}\big(\xi,(\thetav_1,\dots,\thetav_n)\big) \dd\xi\middle|\Gv_n=(\thetav_1,\dots,\thetav_n)\right];\quad x>0,(\thetav_1,\dots,\thetav_n)\in\Gc_n,
\end{cases}
\end{align}
which also determine the investment policy for the n-th period. By Theorem \ref{thm:verif}, this period-by-period forward iteration is guaranteed to yield a PFPP $\{U_n\}_{n\in\Nb_0}$ and an optimal wealth process $\{X^*_n\}_{n\in\Nb_0}$ with initial wealth $x_0>0$.

%
%
\begin{algorithm}[tb]
\caption{Investment policy according to a PFPP}\label{alg:PFPP}
\begin{algorithmic}
\Statex
\Require{initial wealth $x_0$ and initial inverse marginal $I_0=U_0'$}
\State $\gv\gets[~]$, $X^*_0\gets x_0$
\For{$n=0,1,\dots$}
\State \textbf{Step 1:} Observe $\Thetav_{n+1}$. Set $\gv\gets\gv\oplus\Thetav_{n+1}$ and $\nu\gets$ the distribution of $\rho_{n+1}|_{\Gv_{n+1}=\gv}$.\vspace{1em}
\State \textbf{Step 2:} Find $I_{n+1}\in \Ic$ satisfying $\int_{\Rb_+}I_{n+1}(\rho y)\dd \nu(\rho)<+\infty$ and\\
\hspace{5.2em} $\int_{\Rb_+}\rho I_{n+1}(\rho y)\dd\nu(\rho) = I_n(y)$ for all $y>0$.
This is Problem \ref{prob:PFPPConstructionStep}.\vspace{1em}
\State \textbf{Step 3:} Starting with wealth $X^*_n$, invest over time period $[n,n+1]$
to replicate\\
\hspace{5.2em}  the payoff $X^*_{n+1}:= I_{n+1}\big(\rho_{n+1}I_n^{-1}(X^*_n)\big)$ at $n+1$. This is possible since\\
\hspace{5.2em} the market is complete and $\Eb[Z_{n+1}X^*_{n+1}|\Fc_n]=Z_nX^*_n$.
\EndFor
\end{algorithmic}
\end{algorithm}

Algorithm \ref{alg:PFPP} provides a general procedure for implementing an \emph{investment policy} according to the framework provided by PFPPs. The algorithm is a forward-in-time iteration. It takes the initial wealth $x_0>0$ and the initial inverse marginal $I_0\in\Ic$ as its initial inputs. For each evaluation period $[n,n+1]$, $n\in\Nb_0$, it then performs three tasks sequentially. Firstly, at time $n$, it observes the value of $\Thetav_{n+1}$ (which are assumed to be $\Fc_n$-measurable). Although this step is the most important step, we do not explore it in details. The complexity of this step depend on the type of the market model, and it falls into the broader topic of parameter estimation and machine learning. In general, this step involves calibrating the model (e.g. estimating drift and volatility in the Black-Scholes model) and/or consulting with market experts. Secondly, still at time $n$ (and \emph{after} observing $\Thetav_{n+1}$), the algorithm solves an integral equation to obtain $I_{n+1}$, which is essentially \eqref{eq:IntEq_Stepn} in integral form and for the observed value of $\Gv_{n+1}=(\Thetav_1,\dots,\Thetav_{n+1})$. We spend the rest of the paper solving this integral equation. The third step is a replication problem. Specifically, the $I_{n+1}$ found in the second step determines the optimal wealth $X^*_{n+1}$ which, by  \eqref{eq:IntEq_Stepn}, satisfies $\Eb[Z_{n+1}X^*_{n+1}|\Fc_n]=Z_nX^*_n$. Since we have assumed that the market is arbitrage-free and complete, there is a unique strategy over time period $[n,n+1]$ that, starting from $X^*_n$ at time $n$, replicates $X^*_{n+1}$ at time $n+1$. The specifics of this step depends on the market model. We don't go into the details since it is a well-studied subject in mathematical finance.

Note that we have not included calculations for the PFPP $\{U_n\}_{n\in\Nb_0}$ in Algorithm \ref{alg:PFPP}, as only the inverse marginals $I_n$, $n\in\Nb_0$, are needed for calculating the optimal wealth process (and, thus, obtaining the optimal investment positions). Furthermore, in Step 2 of Algorithm \ref{alg:PFPP}, we only need to solve \eqref{eq:IntEq_Stepn} for \emph{one realization} of the random variable $\Gv_n$, that is for $\Gv_n=\gv$ with $\gv$ obtained in Step 1 of the algorithm. In other words, Algorithm \ref{alg:PFPP} will create only one path of the optimal wealth process. Theorem \ref{thm:verif} guarantees that the wealth trajectories generated by Algorithm \ref{alg:PFPP} correspond to an optimal wealth process of a PFPP with initial utility function $U_0$.

In the next section, we show how to solve \eqref{eq:IntEq_Stepn} for an $I_n$ satisfying Condition $(i)$ of Theorem \ref{thm:verif}. Specifically, we will analyze the following problem which, as we just discussed, is the only remaining step for constructing PFPPs.
\begin{problem}\label{prob:PFPPConstructionStep}
Consider the market setting of Section \ref{sec:Market} with Assumptions \ref{asum:SPD} and \ref{asum:rho} holding. Given an $I_{n-1}:\Rb_+\times\Gc_{n-1}\to \Rb_+$ and the distribution of $\rho_n|_{\Gv_n=\gv}$ for all $\gv\in\Gc_n$, find an $I_n:\Rb_+\times\Gc_n\to \Rb_+$ satisfying Conditions $(i)$ and $(ii)$ of Theorem \ref{thm:verif}.\qed
\end{problem}

%
%

\section{The integral equation}\label{sec:ConvEq}

In this section, we first transform Problem \ref{prob:PFPPConstructionStep} into an integral equation, namely, \eqref{eq:IntEq} below. In Subsection \ref{sec:deconvolution}, we then provide a general approach for solving the integral equation by turning it into a convolution equation and then applying the Fourier transform. Finally, in Subsection \ref{sec:CMIM}, we provide existence and uniqueness of the solution to Problem \ref{prob:PFPPConstructionStep} within a special class of Completely Monotonic Inverse Marginal (CMIM) functions. Our discussion culminates in Theorem \ref{thm:PFPP_CMIM} and Algorithm \ref{alg:PFPP_CMIM} which provide an explicit forward construction for PFPP with inverse marginals that are completely monotonic. 

To ease the notations throughout this section, we ignore notational dependence on $n$ and $(\thetav_1,\dots,\thetav_{n-1})$ which appear on both sides of \eqref{eq:IntEq_Stepn}. For instance, instead of $I_{n-1}\big(y,(\thetav_1,\dots,$ $\thetav_{n-1})\big)$, we use $I_0(y)$. Similarly, we replace  $I_n\big(y,(\thetav_1,\dots,\thetav_{n-1},\thetav)\big)$ with $I_1(y,\thetav)$. We introduce the family of probability measures
\begin{align}\label{eq:nu}
\nu_\thetav\big(B\big):=\Eb\left[\Ib_{\{\rho_n\in B\}} \middle|\Gv_n=(\thetav_1,\dots,\thetav_{n-1},\thetav)\right];
\quad B\in\Bc(\Rb), \thetav\in\Xi,
\end{align}
in which
\begin{align}\label{eq:Xi}
\Xi:=\{\thetav\in\Xi_n:(\thetav_1,\dots,\thetav_{n-1},\thetav) \in\Gc_n\}\subseteq\Rb^{M_n},
\end{align}
and 
$\Bc(\Rb)$ denotes the $\sig$-algebra of the Borel subsets of $\Rb$. Note that $\supp(\nu_\thetav)\subseteq\Rb_+$, since $\rho_n>0$ \as~by Assumption \ref{asum:SPD}.
Note also that
\begin{align}\label{eq:ProbMeasure}
\nu_\thetav(\Rb_+) = 1 = \int_{\Rb_+} \rho \dd \nu_\thetav(\rho);
\quad \thetav\in\Xi.
\end{align}
The first equality holds since $\nu_\thetav$ is a probability measure. The second equality holds since $\Eb\left[\rho_n \middle|\Gv_n\right]=\Eb\left[\rho_n \middle|\Fc_n\right]=1$ by Assumptions \ref{asum:SPD} and \ref{asum:rho}.

Using the above notations, Problem \ref{prob:PFPPConstructionStep} is written in the following simplified form.

\begin{problem}\label{prob:IntEq}
Let $I_0\in\Ic$ be an inverse marginal, $\Xi\subseteq\Rb^M$ be a Borel set for some $M\in\Nb$, and $\{\nu_\thetav\}_{\thetav\in\Xi}$ be a family of measures on $\Rb_+$ satisfying \eqref{eq:ProbMeasure}.
Find a function $I_1:\Rb_+\times\Xi\to\Rb_+$ satisfying
\begin{align}\label{eq:IntEq}
\int_{\Rb_+} \rho I_1(y\rho, \thetav)\dd \nu_\thetav(\rho) = I_0(y);\quad y>0,\thetav\in\Xi,
\end{align}
such that $\int_{\Rb_+} I_1(y\rho, \thetav)\dd \nu_\thetav(\rho)<\infty$ and $I_1(\cdot,\thetav)\in\Ic$ for all $y>0$ and $\thetav\in\Xi$.\qed
\end{problem}\vspace{1ex}

%
%

\subsection{The deconvolution Problem}\label{sec:deconvolution}
We start our analysis of Problem \ref{prob:IntEq} by solving the integral equation \eqref{eq:IntEq}, in which $\nu_\thetav$ and $I_0$ are known and $I_1$ is unknown. By setting $y=\ee^s$, $\rho=\ee^{-t}$, $J_0(t):= I_0(\ee^t)$, and $J_1(t,\thetav) := I_1(\ee^t, \thetav)$, we transform \eqref{eq:IntEq} into
\begin{align}\label{eq:Deconvolution}
\int_\Rb J_1(s-t,\thetav)\dd\nut_\thetav(t) = J_0(s),\quad s\in\Rb, \thetav\in\Xi,
\end{align}
in which $\nut_\thetav$ is the probability measure given by
\begin{align}\label{eq:nut}
\nut_\thetav(B):=\int_{\ee^{-B}} \rho \dd \nu_\thetav(\rho);\quad B\in\Bc(\Rb),\thetav\in\Xi.
\end{align}
Note that $\nut_\thetav(\Rb)=1$ because of \eqref{eq:ProbMeasure}.
The left side of \eqref{eq:Deconvolution} is the convolution $J_1(\cdot,\thetav)\conv\nut_\thetav$. Thus, we obtain the following \emph{deconvolution problem}
\begin{align}\label{eq:ConvEq}
J_1(\cdot,\thetav)\conv \nut_\thetav = J_0;\quad \thetav\in \Xi,
\end{align}
in which $J_0$ and $\nut_\thetav$ are known and $J_1$ is unknown.

Deconvolution problems are, in general, difficult to solve. Their solution may not exists or may not be unique. The general approach for solving \eqref{eq:ConvEq} is to exploit the \emph{convolution theorem} which, loosely speaking, states that for ``sufficiently regular'' functions $f$ and $g$, one has $\F[f\conv g] = \F[f]\F[g]$, in which $\F$ is the Fourier transform $\F[f](s)= \int_\mathbb{R} e^{-\ii st} f(t)\dd t$, $s\in\Rb$. To formally solve \eqref{eq:ConvEq}, we take the Fourier transform of both sides and then apply the convolution theorem to obtain
\begin{align}
&\F[J_0]=\F\big[J_1(\cdot,\thetav)\conv\nut_\thetav\big]=\F\big[J_1(\cdot,\thetav)\big] \F[\nut_\thetav]\\
\label{eq:FormalSol}
&\quad\Longrightarrow\quad
J_1(\cdot,\thetav) = \F^{-1}\left[\frac{\F[J_0]}{\F[\nut_\thetav]}\right]=J_0\conv\F^{-1}\left[\frac{1}{\F[\nut_\thetav]}\right],
\end{align}
for $\thetav\in\Xi$, in which $\F^{-1}[g](t)=(2\pi)^{-1}\int_\Rb \ee^{\ii st}g(s)\dd s$ is the inverse Fourier transform. Since we have assumed $J_1(t,\thetav) = I_1(\ee^t, \thetav)$, we obtain that $I_1(y,\thetav)=J_1(\log y, \thetav)$, $(y,\thetav)\in\Rb_+\times\Xi$, satisfies \eqref{eq:IntEq}. With $I_1$ at hand, we can then check if the remaining requirements in Problem \ref{prob:IntEq} are satisfied. If so, we have found a solution. 

The heuristic argument represented by \eqref{eq:FormalSol} is flawed however. Firstly, it assumes that $J_0(t)=I_0(\ee^t)$, $t\in\Rb$, has a Fourier transform. This assumption fails even for the simple case of power utility $U(x)=\frac{x^{1-\gam}-1}{1-\gam}$, $x,\gam>0$. For this case, $I_0(y)=U'^{(-1)}(y)=y^{-1/\gam}$, $y>0$, and the improper integral $\int_\Rb \ee^{-\ii st}J_0(t)\dd t= \int_\Rb \ee^{-\frac{1}{\gam}t-\ii st}\dd t$ is divergent.
Secondly, the convolution theorem and the convolution operator on the left side of \eqref{eq:ConvEq} require that either $J_1(\cdot,\thetav)$ or $\nut_\thetav$ has a compact support, which is not true in general. In fact, $J_1(t,\thetav)= I_1(\ee^{t},\thetav)$, $t\in\Rb$, cannot have compact support because \eqref{eq:Ic} requires that $I_1(y,\thetav)>0$ for $y>0$. Thus, one could only assume that $\nut_\thetav$ has compact support. While such an assumption holds for some scenarios (say, the binomial market, see Subsection \ref{sec:binom}), it fails in other cases where $\supp(\rho_n|_{\Gv_n})$ is not compact. For instance, in the Black-Scholes model, $\rho_n|_{\Gv_n}$ has a log-normal distribution and, thus, $\supp(\nu_\thetav)=\Rb_+$ and is not compact (see Subsection \ref{sec:BS}).

Our next result, namely, Proposition \ref{prop:ExistenceUniquness} below, establishes the existence and uniqueness of the solution to the deconvolution problem \eqref{eq:Deconvolution} under additional regularity conditions on $\nu_\thetav$, $J_0$, and $J_1$. These conditions and the proof of Proposition \ref{prop:ExistenceUniquness} rely on certain facts from the theory of distributions (also known as generalized functions) and the Fourier transform for tempered distributions. For the sake of completeness and ease of reference, a brief review has been included in Appendix \ref{sec:FourierAnalysis}. Further details can be found in most texts on the Fourier analysis, for instance, \cite{Hormander1990}.

The following assumption is our main regularity assumption on the measures $\nu_\thetav$, $\thetav\in\Xi$. In its statement, $\Sc'$ is the space of tempered distributions (see Definition \ref{def:TempDist}), $\F$ denotes the Fourier transform on $\Sc'$ (see Definition \ref{def:Fourier_tempered}), and $\Cc^{\infty}$ denote the set of all complex-valued infinitely-differentiable functions with domain $\Rb$.

\begin{assumption}\label{assum:gam1gam2}
There exist constants $0<\gam_1\le\gam_2$ such that, for $\thetav\in\Xi$ and $k\in\{1,2\}$, the $\sig$-finite Borel measure $\mu_{\thetav,k}$ given by
\begin{align}\label{eq:mu_thetav}
\mu_{\thetav,k}(B) := \int_{\ee^{-B}} \rho^{1-\frac{1}{\gam_k}} \dd \nu_\thetav(\rho)=\int_B \ee^{\frac{t}{\gam_k}}\dd\nut_\thetav(t);\quad B\in\Bc(\Rb),
\end{align}
satisfy $\mu_{\thetav,k}\in\Sc'$ and $\F[\mu_{\thetav,k}]\in\Cc^\infty$. Here, $\nut_\thetav$ is given by \eqref{eq:nut}.\qed
\end{assumption}

Our main regularity conditions on $J_0(\cdot)$ and the solution $J_1(\cdot,\thetav)$, $\thetav\in\Xi$, is that they belong to the following set,
\begin{align}\label{eq:Jgam1gam2}
\Jc(\gam_1,\gam_2):=
\left\{ J\in L^1_\text{loc}:
t\mapsto\left(\ee^{\frac{1}{\gam_1}t}\1_{\{t<0\}}
+\ee^{\frac{1}{\gam_2}t}\1_{\{t\ge0\}}\right)|J(t)|\in\Sc'
\right\},
\end{align}
with $0<\gam_1\le\gam_2$ as in Assumption \ref{assum:gam1gam2}. Here, $L^1_\text{loc}$ denotes the set of real-valued functions that are integrable on compact subsets of $\Rb$.

The next result, which is our second main result, provides conditions for existence and uniqueness of a solution $J_1$ to the deconvolution problem \eqref{eq:Deconvolution} satisfying $J_1(\cdot,\thetav)\in\Jc(\gam_1,\gam_2)$, for all $\thetav\in\Xi$.

\begin{proposition}\label{prop:ExistenceUniquness}
Assume that $\{\nu_\thetav\}_{\thetav\in\Xi}$ satisfy Assumption \ref{assum:gam1gam2} for some constants $0<\gam_1\le\gam_2$ and let the Borel measures $\mu_{\thetav,k}$, $\thetav\in\Xi$, $k\in\{1,2\}$, be as in \eqref{eq:mu_thetav}. 
Let $J_0\in\Jc(\gam_1,\gam_2)$ (with $\Jc(\gam_1,\gam_2)$ as in \eqref{eq:Jgam1gam2}) and, for $t\in\Rb$, define $J_{0,1}(t):=J_0(t)\ee^{\frac{1}{\gam_1}t}\1_{\{t<0\}}$ and $J_{0,2}(t):=J_0(t)\ee^{\frac{1}{\gam_2}t}\1_{\{t\ge0\}}$.
Assume further that
the following conditions hold for all $\thetav\in\Xi$ and $k\in\{1,2\}$:\vspace{1ex}\\
\noindent$(i)$ $\F[J_{0,k}]/\F[\mu_{\thetav,k}]\in\Sc'$,\\
\noindent$(ii)$ $J_{1,k}(\cdot,\thetav):=\F^{-1}\big[\F[J_{0,k}]/\F[\mu_{\thetav,k}]\big]\in L^1_\text{loc}(\Rb)$, and\\
\noindent$(iii)$ $\int_{\Rb} |J_{1,k}(s-t,\thetav)|\dd\mu_{\thetav,k}(t)<+\infty$ for all $s\in\Rb$.\vspace{1ex}\\
Define $J_1(t,\thetav):=\ee^{-\frac{1}{\gam_1}t}J_{1,1}(t,\thetav) + \ee^{-\frac{1}{\gam_2}t}J_{1,2}(t,\thetav)$ for $(t,\thetav)\in\Rb\times\Xi$. Then, $J_1(\cdot,\thetav)\in\Jc(\gam_1,\gam_2)$, $\thetav\in\Xi$, and  $J_1$ is a solution of the deconvolution problem \eqref{eq:Deconvolution} with $\nut_\thetav$ as in \eqref{eq:nut}. Furthermore, for any $\thetav\in\Xi$, if $\F[\mu_{\thetav,k}](\xi)\ne0$ for all $(k,\xi)\in\{1,2\}\times\Rb$, $\Jt\in\Jc(\gam_1,\gam_2)$, and $\int_\Rb\ee^{\frac{s-t}{\gam_k}}\Jt(s-t)\dd\mu_{\thetav,k}(t) = J_{0,k}(s)$, $(k,s)\in\{1,2\}\times\Rb$, then $\Jt=J_1(\cdot,\thetav)$ almost everywhere on $\Rb$.\qed
\end{proposition}

\begin{proof}
See Appendix \ref{app:ExistenceUniquness}.
\end{proof}

\begin{remark}
To motivate introducing $J_{0,k}, k\in\{1,2\}$ in Proposition \ref{prop:ExistenceUniquness}, consider the case that the initial inverse marginal is a (convex) combination of two CRRA inverse marginals, that is
\begin{align}
I_0(y) = \al y^{-\frac{1}{\gam_1}} + (1-\al) y^{-\frac{1}{\gam_2}}, \quad y>0,
\end{align}
in which $0\le \al\le 1$ and $0<\gam_1<\gam_2$ are constants. The integral equation \eqref{eq:IntEq} becomes
\begin{align}
\int_{\Rb_+} \rho I_1(y\rho, \thetav)\dd \nu_\thetav(\rho) = \al y^{-\frac{1}{\gam_1}} + (1-\al) y^{-\frac{1}{\gam_2}};\quad y>0,\thetav\in\Xi.
\end{align}
To solve this equation, we can exploit the fact that the integral equation is linear and try the ansatz $I_1 = \al I_{1,1} + (1-\al)I_{1,2}$. Indeed going down this path would lead to the approach used for the completely monotonic case in Section \ref{sec:CMIM}. Here, we will not pursue this argument because our purpose is to justify the approach taken in Proposition \ref{prop:ExistenceUniquness}, which is for the more general case in which the initial data $I_0$ is not necessarily completely monotonic.

Let us define $y=\ee^s$, $\rho=\ee^{-t}$, and $J_1(t,\thetav) := I_1(\ee^t, \thetav)$ to obtain
\begin{align}
(J_1(\cdot,\thetav)\conv\nut_\thetav)(s) = \int_\Rb J_1(s-t,\thetav)\dd \nut_\thetav(t) = \al \ee^{-\frac{1}{\gam_1}s} + (1-\al) \ee^{-\frac{1}{\gam_2}s}=:J_0(s);\quad s\in\Rb,\thetav\in\Xi,
\end{align}
with $\nut_\thetav$ as in \eqref{eq:nut}. To use the convolution theorem, we would like to apply the Fourier transform to $J_0$. $\F[J_0]$ is not defined since $J_0(t)$ behaves like $\ee^{-t/\gam_1}$ as $t\to -\infty$ and, thus, is not a tempered distribution. The general approach to fix this is to multiply $J_0$ by the exponential function $\ee^{t/\gam_1}$. Doing so, however, will spoil the asymptotic behavior on the other end (i.e. as $t\to+\infty$), since $\ee^{t/\gam_1}J_0(t)$ behave as $\ee^{\left(\frac{1}{\gam_1}-\frac{1}{\gam_2}\right)t}$ for $t\to+\infty$. To circumvent this, we can multiply by the function $\ee^{t/\gam_1}\Ib_{\{t<0\}}$. Doing so has one disadvantage, $\ee^{t/\gam_1}\Ib_{\{t<0\}}J_0(t)$ is zero for $t>0$. To preserve the function on the interval $(0,+\infty)$, we can instead multiply by $\ee^{t/\gam_1}\Ib_{\{t<0\}}+\ee^{t/\gam_2}\Ib_{\{t\ge 0\}}$. Note that $J_0(t)\left[\ee^{t/\gam_1}\Ib_{\{t<0\}}+\ee^{t/\gam_2}\Ib_{\{t\ge 0\}}\right]$ is now a bounded function and thus a tempered distribution. Following this approach yields the argument in the proof of Proposition \ref{prop:ExistenceUniquness} in Appendix \ref{app:ExistenceUniquness}.\qed
\end{remark}

We end this section by an example in which the solution of the deconvolution problem \eqref{eq:Deconvolution} is not unique.

\begin{example}\label{ex:AZZ}
The deconvolution problem \eqref{eq:Deconvolution} may have non-unique solutions. For instance, let $\nut_\thetav=\beta\del_{-\al}+(1-\beta)\del_0$ for constants $\al>0$ and $\beta\in(0,1)$. Then, \eqref{eq:Deconvolution} becomes the functional equation
\begin{align}\label{eq:Deconv_AZZ}
\beta J_1(s+\al)+(1-\beta) J_1(s)=J_0(s),\quad s\in\Rb.
\end{align}
Assume that $J_1$ is a solution of this equation, and define
\begin{align}\label{eq:Deconv_AZZ_sol}
\Jt(t):=J_1(t)+\left(\frac{1-\beta}{\beta}\right)^{t/\al}\psi\left(\frac{\pi t}{\al}\right),\quad t\in\Rb,
\end{align}
in which $\psi$ is an anti-periodic function satisfying $\psi(t+\pi)=-\psi(t)$, $t\in\Rb$. For instance, we may choose $\psi=M\sin(t)$ for a constant $M\ne0$.
For $s\in\Rb$, we have that
\begin{align}
&\beta\Jt(s+\al)+(1-\beta)\Jt(s)\\
&= \beta J_1(s+\al)+(1-\beta) J_1(s) + \beta\left(\frac{1-\beta}{\beta}\right)^{1+\frac{s}{\al}}\sin\left(\frac{\pi s}{\al} +\pi\right) + \beta\left(\frac{1-\beta}{\beta}\right)^{1+\frac{s}{\al}}\sin\left(\frac{\pi s}{\al}\right)\\
&= J_0(s).
\end{align}
Thus, the solution of \eqref{eq:Deconv_AZZ} is not unique.

Let us confirm that the uniqueness assertion in Proposition (i.e. its last statement) is consistent with this example. Assume that $0<\gam_1\le\gam_2$ are such that $\F[\mu_{\thetav,k}](\xi)\ne0$ for all $(k,\xi\in\{1,2\}\times\Rb$, in which $\mu_{\thetav,k}$ are given by \eqref{eq:mu_thetav}, namely,
\begin{align}
\mu_{\thetav,k} = \beta\ee^{-\frac{\al}{\gam_k}}\del_{-\al}+(1-\beta)\del_0.
\end{align}
Since $0<\beta<1$ and
\begin{align}\label{eq:Fmu_Azz}
\F[\mu_{\thetav,k}](\xi) = \F[\beta\ee^{-\frac{\al}{\gam_k}}\del_{-\al}+(1-\beta)\del_0] = \beta\ee^{-\frac{\al}{\gam_k} + \ii\al\xi} + 1-\beta, \quad \xi\in\Rb,
\end{align}
it follows that $\F[\mu_{\thetav,k}](\xi)\ne0$ for all $(k,\xi)\in\{1,2\}\times\Rb$ if and only if $\frac{1}{\gam_k}+\frac{1}{\al}\log\left(\frac{1-\beta}{\beta}\right)\ne 0$, $k\in\{1,2\}$. For such values of $\gam_1$ and $\gam_2$, we have that $\Jt\notin\Jc(\gam_1,\gam_2)$. Indeed, \eqref{eq:Deconv_AZZ_sol} yields that
\begin{align}
\ee^{\frac{t}{\gam_k}}\Jt(t) 
= \ee^{\frac{t}{\gam_k}}J_1(t) + \ee^{t\left[\frac{1}{\gam_k}+\frac{1}{\al}\log\left(\frac{1-\beta}{\beta}\right)\right]}
\psi\left(\frac{\pi t}{\al}\right),\quad t\in\Rb,k\in\{1,2\}.
\end{align}
Therefore, $t\mapsto\left(\ee^{\frac{1}{\gam_1}t}\1_{\{t<0\}}
+\ee^{\frac{1}{\gam_2}t}\1_{\{t\ge0\}}\right)|\Jt(t)|$ cannot be a tempered distribution since it has exponential growth as either $t\to+\infty$ or $t\to-\infty$ depending on the sign of $\frac{1}{\gam_k}+\frac{1}{\al}\log\left(\frac{1-\beta}{\beta}\right)$. In short, as long as we require $\frac{1}{\gam_k}+\frac{1}{\al}\log\left(\frac{1-\beta}{\beta}\right)\ne 0$, $k\in\{1,2\}$, then the solution of \eqref{eq:Deconv_AZZ} is unique in the set $J(\gam_1,\gam_2)$, as stated by Proposition \ref{prop:ExistenceUniquness}.

The only case that the non-unique solutions $\Jt$ given by \eqref{eq:Deconv_AZZ_sol} belong to the set $J(\gam_1,\gam_2)$ is when, for at least one $k'\in\{1,2\}$, we have that $\frac{1}{\gam_{k'}}+\frac{1}{\al}\log\left(\frac{1-\beta}{\beta}\right)= 0$. In this case, \eqref{eq:Fmu_Azz} yields that $F[\mu_{\thetav,k'}](\xi)=(1-\beta)(1+\ee^{\ii\al\xi})$, $\xi\in\Rb$. In particular, $F[\mu_{\thetav,k'}](\pm2m\pi/\al)=0$, $m\in\Nb$. Thus, at least one of the assumptions of Proposition \ref{prop:ExistenceUniquness} is not satisfied and, as expected, the proposition does not apply.\qed
\end{example}

%
%
\subsection{Completely monotonic inverse marginals}\label{sec:CMIM}

In the previous section, we focused on the integral equation \eqref{eq:IntEq} and derived rather technical existence and uniqueness conditions for its solution. In this section, we consider the more general Problem \ref{prob:IntEq}. However, we restrict our attention to  solutions of this problem within a special subclass of inverse marginal functions, namely, completely monotonic inverse marginal (CMIM) functions. Doing so enables us to provide more explicit solutions that are easier to interpret.

We start by defining CMIM functions. See \cite{Kallblad2020} and \cite{MostovyiSirbuZariphopoulou2020}, among others, for a more detailed discussion on CMIM functions and historical insights.

\begin{definition}\label{def:CMIM}
For a finite Borel measure $m$ with support in $\Rb_+$, a function $I:\Rb_+\to\Rb_+$ is a \emph{completely monotonic inverse marginal (CMIM)} function with risk-aversion measure $m$ if
\begin{align}\label{eq:CMIM}
I(y) = \int_{\Rb_+} y^{-\frac{1}{\gam}}\dd m(\gam);\quad y>0,
\end{align} 
in which it is assumed that the right side is absolutely integrable for all $y>0$. For any constants $0<\gam_1\le\gam_2$, we denote by $\CMIM(\gam_1,\gam_2)$ the set of all CMIM functions with a risk-aversion measure $m$ that has compact support in $(\gam_1,\gam_2)$, i.e. $\supp(m)\subset(\gam_1,\gam_2)$.\qed
\end{definition}

\begin{remark}
Note that our definition of CMIM functions is more restricted than the one in the literature (e.g. Definition 3.8 of \cite{Kallblad2020} and Definition 4.1 of \cite{MostovyiSirbuZariphopoulou2020}). In particular, we assume that the measure $m$  has compact support. This assumption is adapted to simplify the proof of the results that follow. It can be relaxed but at the expense of strengthening the assumptions on the measure $\nu_\thetav$. For instance, in the case where $\nu_\thetav$ is the log-normal density, \cite{Kallblad2020} solves the integral equation for general CMIM. Note also that any CMIM utility function that behave like a power utility function for very small and very large values of wealth are included in our definition of CMIM (the power utilities on the two end can be different).
\end{remark}

The following lemma provides a basic property of CMIM functions, namely, that every $\CMIM$ function is an inverse marginal function.
\begin{lemma}\label{lem:CMIM}
$\CMIM(\gam_1,\gam_2)\subset\Ic$ for all $0<\gam_1\le\gam_2$.\qed
\end{lemma}
\begin{proof}
Let $I(y)=\int_{\Rb_+} y^{-\frac{1}{\gam}}\dd m(\gam)$, $y>0$, in which $m$ is a finite Borel measure with $\supp(m)\subseteq(\gam_1,\gam_2)$. For $\gam\in(\gam_1,\gam_2)$, we have that $y^{-1/\gam_1}\le y^{-1/\gam}\le y^{-1/\gam_2}$ for $y\ge1$ and $y^{-1/\gam_2}\le y^{-1/\gam}\le y^{-1/\gam_1}$ for $0<y\le1$. Therefore, the dominated convergence theorem yields that $I\in\Cc^1(\Rb_+)$ and that $I'(y)=-\frac{1}{\gam}\int_{\Rb_+} y^{-\frac{1+\gam}{\gam}}\dd m(\gam)<0$ for $y>0$. Furthermore, 
\begin{align}\label{eq:CMIM_bound1}
y^{-\frac{1}{\gam_1}}\le \frac{1}{m\big((\gam_1,\gam_2)\big)} I(y) \le y^{-\frac{1}{\gam_2}},\quad y\ge1,
\end{align}
and
\begin{align}\label{eq:CMIM_bound2}
y^{-\frac{1}{\gam_2}}\le \frac{1}{m\big((\gam_1,\gam_2)\big)} I(y) \le y^{-\frac{1}{\gam_1}},\quad 0<y\le1.
\end{align}
From \eqref{eq:Ic}, it then follows that $I\in\Ic$.
\end{proof}

Next, we state the third main result of our paper. It shows that, under a mild integrability condition on measure $\nu_\thetav$ (namely, \eqref{eq:CMIM_nu} below), if $I_0$ is a CMIM function, then there is a unique solution $I_1$ of Problem \eqref{prob:IntEq} such that $I_1(\cdot,\thetav)$, $\thetav\in\Xi$, is a CMIM function. Furthermore, $I_1$ is explicitly given by \eqref{eq:IntEq_CMIM_sol}.

\begin{theorem}\label{thm:CMIM}
In Problem \ref{prob:IntEq}, assume that there exist constants $0<\gam_1\le\gam_2$ such that
\begin{align}\label{eq:CMIM_nu}
\int_{\Rb_+} \left(\rho^{-\frac{1}{\gam_1}}+\rho^{1-\frac{1}{\gam_1}}+\rho^{1-\frac{1}{\gam_2}}\right)\dd\nu_\thetav(\rho)<+\infty,
\quad \thetav\in\Xi.
\end{align}
Assume further that $I_0\in\CMIM(\gam_1,\gam_2)$ and, in particular, that 
$I_0(y)=\int_{\gam_1}^{\gam_2} y^{-1/\gam}\dd m_0(\gam)$, $y>0$, for a finite Borel measure $m_0$ such that $\supp(m_0)\subset(\gam_1,\gam_2)$. Then,
\begin{align}\label{eq:IntEq_CMIM_sol}
I_1(y,\thetav) := \int_{\gam_1}^{\gam_2} y^{-\frac{1}{\gam}}\left(
\int_{\Rb_+} \rho^{1-\frac{1}{\gam}}\dd\nu_\thetav(\rho)
\right)^{-1}\dd m_0(\gam);\quad (y,\thetav)\in\Rb_+\times\Xi,
\end{align}
is the unique solution of Problem \ref{prob:IntEq} satisfying $I_1(\cdot,\thetav)\in\CMIM(\gam_1,\gam_2)$, $\thetav\in\Xi$.\qed
\end{theorem}
\begin{proof}
See Appendix \ref{sec:CMIM_proof}.
\end{proof}

We end this section by providing existence and uniqueness conditions for PFPPs whose inverse marginals are CMIM functions. The result follows directly from combining Theorem \ref{thm:verif} and Theorem \ref{thm:CMIM}, therefore, we omit its proof.

\begin{theorem}\label{thm:PFPP_CMIM}
Consider the market setting of Section \ref{sec:Market} with Assumptions \ref{asum:SPD} and \ref{asum:rho} holding. Assume that there exists constants $0<\gam_1\le\gam_2$ such that
\begin{align}\label{eq:CMIM_PFPP_condition}
\Eb\left[\rho_n^{-\frac{1}{\gam_1}}+\rho_n^{1-\frac{1}{\gam_1}}+\rho_n^{1-\frac{1}{\gam_2}}\middle|\Gv_n=\gv\right]<+\infty;
\quad n\in\Nb,\gv\in\Gc_n,
\end{align}
and let $I_0(y):=\int_{\gam_1}^{\gam_2} y^{-1/\gam}\dd m_0(\gam)$, $y>0$, for a finite Borel measure $m_0$ such that $\supp(m_0)\subset(\gam_1,\gam_2)$. Define the finite Borel measures $m_{n,\gv}$, $n\in\Nb$, $\gv\in\Gc_n$, by the iteration
\begin{align}\label{eq:m_n_CMIM}
m_{n,\gv}(B) = \int_B \left(\Eb\left[\rho_n^{1-\frac{1}{\gam}}\middle|\Gv_n=\gv\right]\right)^{-1} \dd m_{n-1,\gv'}(\gam);\quad
B\in\Bc(\Rb), n\in\Nb, \gv\in\Gc_n,
\end{align}
in which $\gv=\gv'\oplus\thetav\in\Gc_n$ such that $\gv'\in\Gc_{n-1}$ and $\thetav\in\Xi_n$,\footnote{For $n=1$, \eqref{eq:m_n_CMIM} becomes $m_{1,\thetav}(B) = \int_B \left(\Eb\left[\rho_1^{1-\frac{1}{\gam}}\middle|\Thetav_1=\thetav\right]\right)^{-1} \dd m_0(\gam)$, $B\in\Bc(\Rb)$ and $\thetav\in\Xi_1$.}
and let
\begin{align}\label{eq:In_CMIM}
I_n(y,\gv) := \int_{\gam_1}^{\gam_2} y^{-\frac{1}{\gam}}\dd m_{n,\gv}(\gam);\quad n\in\Nb, \gv\in\Gc_n.
\end{align}
Then, the unique PFPP $\left\{U_n\right\}_{n\in\Nb_0}$ satisfying $I_0:=U_0^{\prime-1}$ and $U_n^{\prime-1}(\cdot,\gv)\in\CMIM(\gam_1,\gam_2)$, $n\in\Nb, \gv\in\Gc_n$, is given by
\begin{align}\label{eq:U_CMIM}
U_n(x,\gv) &:= U_{n-1}\left(I_{n-1}(1,\gv'),\gv'\right) +\Eb\left[\int_{I_n\left(\rho_n, \gv\right)}^x I_n^{-1}(\xi,\gv) \dd\xi\middle|\Gv_n=\gv\right];\quad
x\in\Rb_+, n\in\Nb, \gv\in\Gc_n,
\end{align}
in which $\gv=\gv'\oplus\thetav\in\Gc_n$ (such that $\gv'\in\Gc_{n-1}$ and $\thetav\in\Xi_n$).\footnote{For $n=1$, \eqref{eq:U_CMIM} becomes $U_1(x,\thetav) := U_0\big(I_0(1)\big) +\Eb\left[\int_{I_1\left(\rho_1,\thetav\right)}^x I_1^{-1}(\xi,\thetav) \dd\xi\middle|\Thetav_1=\thetav\right]$, $(x,\thetav)\in \Rb_+\times \Xi_1$.}
A corresponding corresponding optimal wealth process starting with initial wealth $x_0>0$ is given by
\begin{align}\label{eq:XStar_CMIM}
X^*_n:=I_n\left(\rho_n I_{n-1}^{-1}(X^*_{n-1},\Gv_{n-1}), \Gv_n\right);\quad n\in\Nb,
\end{align}
with $X^*_0=x_0$.\qed
\end{theorem}

Algorithm \ref{alg:PFPP_CMIM} provides an investment policy within the framework of a PFPP whose inverse marginals are CMIM functions. The algorithm has the same general structure as Algorithm \ref{alg:PFPP} and the discussion at the end of Section \ref{sec:PFPP_construction} still applies. Since the inverse marginals are assumed to be completely monotonic, we can exploit Theorem \ref{thm:CMIM} to make Step 2 more explicit than its counterpart in Algorithm \ref{alg:PFPP}.  

%
%
\begin{algorithm}[t]
\caption{Investment policy according to a PFPP with CMIM functions}\label{alg:PFPP_CMIM}
\begin{algorithmic}
\Statex
\Require{$0<\gam_1\le\gam_2$ satisfying \eqref{eq:CMIM_PFPP_condition}. Initial wealth $x_0>0$.}
\Require{Initial risk-aversion measure $m_0$ satisfying $\supp(m_0)\subset(\gam_1,\gam_2)$.}
\State $\gv\gets[~]$, $X^*_0\gets x_0$, $I_0(y)\gets\int_{\gam_1}^{\gam_2} y^{-1/\gam}\dd m_0(\gam)$.
\For{$n=0,1,\dots$}
\State \textbf{Step 1:} Observe $\Thetav_{n+1}$. Set $\gv\gets\gv\oplus\Thetav_{n+1}$ and $\nu\gets$ the distribution of $\rho_{n+1}|_{\Gv_{n+1}=\gv}$.\vspace{1em}
\State \textbf{Step 2:} $I_{n+1}(y)\gets\int_{\gam_1}^{\gam_2} y^{-1/\gam}\dd m_{n+1}(\gam)$ in which $m_{n+1}$ is a measure equivalent to $m_n$
\\
\hspace{5.2em} with the Radon–Nikodym derivative $\frac{\dd m_{n+1}}{\dd m_n}(\gam)=\left(
\int_{\Rb_+} \rho^{1-\frac{1}{\gam}}\dd\nu(\rho)\right)^{-1}$ for $\gam\in(\gam_1,\gam_2)$.\vspace{1em}
\State \textbf{Step 3:} Starting with wealth $X^*_n$, invest over time period $[n,n+1]$
to replicate\\
\hspace{5.2em}  the payoff $X^*_{n+1}:= I_{n+1}\big(\rho_{n+1}I_n^{-1}(X^*_n)\big)$ at $n+1$. This is possible since\\
\hspace{5.2em} the market is complete and $\Eb[Z_{n+1}X^*_{n+1}|\Fc_n]=Z_nX^*_n$.
\EndFor
\end{algorithmic}
\end{algorithm}

%
%

\section{Examples}\label{sec:Examples}
In our last section, we apply the results of Sections \ref{sec:PFPP} and \ref{sec:ConvEq} in two concrete examples. The first one considers the binomial market of Example \ref{ex:Binom}. Existence and properties of PFPPs in the binomial market have been extensively studied in \cite{AZZ20}, \cite{StrubZhou2021}, and \cite{LiangStrubWang2021}, and we have included this example for comparison with our more general results. In the second example, we construct PFPPs in the generalized Black-Scholes market of Example \ref{ex:BS} which, to the best of our knowledge, is a new result. 

%
%
\subsection{PFPPs in a generalized binomial model}\label{sec:binom}

We start by adapting the general notations used in Sections \ref{sec:PFPP} and \ref{sec:ConvEq} to the binomial market setting of Example \ref{ex:Binom} and formulate Problem \ref{prob:IntEq}. We then focus on the solution of the integral equation \eqref{eq:IntEq} by applying Proposition \ref{prop:ExistenceUniquness}. Finally, by applying Theorems \ref{thm:CMIM} and \ref{thm:PFPP_CMIM} to the binomial market, we provide a construction procedure for PFPPs whose inverse marginal functions are completely monotonic.

As discussed in Remark \ref{rem:Binom_BS_rho}, for this model $\Thetav_n=\{(u_m,d_m,p_m)\}_{m=1+(n-1)N}^{nN}$ and $\Gv_n:=(\Thetav_1,\dots,\Thetav_n)=\{(u_m,d_m,p_m)\}_{m=1}^{nN}$ for $n\in\Nb$. The assumptions on $(u_n,d_n,p_n)$ in Example \ref{ex:Binom} yield that
\begin{align}\label{eq:Xi_n_binom}
\Xi_n:=\supp(\Thetav_n)=\left\{\{(u_m,d_m,p_m)\}_{m=1+(n-1)N}^{nN}:u_m>1,d_m,p_m\in(0,1)\right\}
\intertext{and}
\label{eq:Gc_n_binom}
\Gc_n:=\supp(\Gv_n) = \left\{\{(u_m,d_m,p_m)\}_{m=1}^{nN}:u_m>1,d_m,p_m\in(0,1)\right\},
\end{align}
in which we have abused the notation by using $(u_m,d_m,p_m)$ for the values taken by the random variables $(u_m,d_m,p_m)$. 
Defined the index set
\begin{align}\label{eq:An}
A_n:=\{m\in\Nb: 1+(n-1)N\le m\le nN\};\quad n\in\Nb,
\end{align}
and recall that $q_m:=(1-d_m)/(u_m-d_m)$, $m\in\Nb$, is the conditional risk-neutral probability of upward jump in period $\left[\frac{m-1}{N}, \frac{m}{N}\right]$. Let\footnote{Here, we have abused the notation since we have defined earlier $\rho_n:=Z_n/Z_{n-1}$.}
\begin{align}
\label{eq:rho_n_binom}
\rho_n(\thetav,S) := \prod_{m\in S} \frac{q_m}{p_m} \prod_{m'\in A_n\backslash S} \frac{1-q_{m'}}{1-p_{m'}},
\end{align}
for $n\in\Nb$, $\thetav=\{(u_m,d_m,p_m)\}_{m\in A_n}\in\Xi_n$, and $S\subseteq A_n$.
In light of \eqref{eq:Rho_binom}, $\rho_n(\thetav,S)$ is the value of $\rho_n:=Z_n/Z_{n-1}$ assuming that $\thetav=\{(u_m,d_m,p_m)\}_{m\in A_n}$ are the market parameters for time period $[n-1,n]$ (i.e. $\Thetav_n=\thetav$), that for all $m\in S$ the price has jumped up in the period $\left[\frac{m-1}{N}, \frac{m}{N}\right]$ (i.e. $B_m=1$), and that for all $m'\in A_n\backslash S$ the prices has jumped down in period $\left[\frac{m'-1}{N}, \frac{m'}{N}\right]$ (i.e. $B_{m'}=0$).
Finally, for $n\in\Nb$, $\thetav=\{(u_m,d_m,p_m)\}_{m\in A_n}\in\Xi_n$, and $S\subseteq A_n$, we define
\begin{align}\label{eq:pi_binom}
\pi_n(\thetav,S)&:=\Pb\Big(\rho_n=\rho_n(\thetav, S)\Big|\Thetav_n=\thetav\Big)
=\prod_{m\in S} p_m \prod_{m'\in A_n\backslash S}(1-p_{m'}),
\end{align}
in which we have used \eqref{eq:Binom_p} for the last step.

Let us first consider Problem \ref{prob:IntEq}. By ignoring notational dependence on $n$ and $(\thetav_1,\dots,\thetav_{n-1})$,  let $\Xi := \Xi_n$, $A:=A_n$, $\rho(\cdot,\cdot)=\rho_n(\cdot,\cdot)$, and $\pi(\cdot,\cdot)=\pi_n(\cdot,\cdot)$. From \eqref{eq:Rho_binom}, it follows $\rho_n|_{\Gv_n}=\rho_n|_{\Thetav_n}$. From \eqref{eq:nu},  the family of probability measures $\{\nu_\thetav\}_{\thetav\in\Xi}$ have the following representation
\begin{align}\label{eq:nu_binom}
\nu_\thetav(B) = \sum_{S\subseteq A} \pi(\thetav,S) \del_{\rho(\thetav,S)}(B);\quad (\thetav,B)\in\Xi_n\times\Bc(\Rb),
\end{align}
in which $\del_a$ is the Dirac measure concentrated at $a\in\Rb$. In particular, $\nu_\thetav$  has finite support (specifically, the number of elements of $\supp(\nu_\thetav)$ is at most $2^N$ which is the number of subsets of $A$).

With the above notations, we find the following more explicit form of Problem \eqref{prob:IntEq} in the binomial setting.
\begin{problem}\label{prob:IntEq_binom}
Given an $I_0\in\Ic$, find an $I_1:\Rb_+\times\Xi\to\Rb_+$ such that $I_1(\cdot,\thetav)\in\Ic$ and
\begin{align}\label{eq:IntEq_binom}
\sum_{S\subseteq A} \pi(\thetav,S)\rho(\thetav,S) I_1\big(y\rho(\thetav,S), \thetav\big) = I_0(y),
\end{align}
for all $y>0$ and $\thetav\in\Xi$.\qed
\end{problem}

Note that the third requirement of Problem \ref{prob:IntEq} (i.e. $\int_{\Rb_+} I_1(y\rho, \thetav)\nu_\thetav(\dd \rho)<\infty$) is automatically satisfied since 
$\int_{\Rb_+} I_1(y\rho, \thetav)\nu_\thetav(\dd \rho)=
\sum_{S\subseteq A_n} p(\thetav,S)I_1\big(y\rho(\thetav,S), \thetav\big)$ and the summation on the right side has finite number of terms.

\begin{remark}
By setting $N=1$, \eqref{eq:IntEq_binom} becomes the functional equation in \cite{AZZ20} (that is, equation (5.4) on page 335 therein). More generally, \eqref{eq:IntEq_binom} is equivalent to the functional equation in \cite{LiangStrubWang2021} (see, equation (8) on page 18 therein).\qed
\end{remark}

Next, we show that Assumption \ref{assum:gam1gam2} is true in the binomial setting. For arbitrary constants $0<\gam_1\le\gam_2$ and for $\thetav\in\Xi$, define the measure $\mu_{\thetav,k}(B) := \int_{\ee^{-B}} \rho^{1-\frac{1}{\gam_k}} \dd \nu_\thetav(\rho)$, $B\in\Bc(\Rb)$, $k\in\{1,2\}$. From \eqref{eq:nu_binom}, we obtain that
\begin{align}\label{eq:mu_binom}
\mu_{\thetav,k}  
=\sum_{S\subseteq A} 
\big(\rho(\thetav,S)\big)^{1-\frac{1}{\gam_k}}
\pi(\thetav,S) 
\del_{-\log\rho(\thetav,S)};\quad k\in\{1,2\}.
\end{align}
Since there are only finite number of terms in the sum on the right side, we have $\mu_{\thetav,k}\in\Ec'\subset\Sc'$ which, in turn, yields that $\F[\mu_{\thetav,k}]\in\Cc^{\infty}$ by Lemma \ref{lem:FourierLaplace}. Indeed, we can explicitly calculate
\begin{align}\label{eq:Fmu_binom}
\F[\mu_{\thetav,k}](\xi)
=\sum_{S\subseteq A} 
\big(\rho(\thetav,S)\big)^{1-\frac{1}{\gam_k}}
\pi(\thetav,S)
\ee^{\ii\log\big(\rho(\thetav,S)\big)\xi};\quad \xi\in\Rb.
\end{align}
We have shown that Assumption \ref{assum:gam1gam2} holds in the binomial market for any choice of $0<\gam_1\le\gam_2$.

Next, we analyze the functional equation \eqref{eq:IntEq_binom}. By setting $y=\ee^s$, $J_0(t):= I_0(\ee^t)$, and $J_1(t,\thetav) := I_1(\ee^t, \thetav)$, we transform \eqref{eq:IntEq_binom} into the following equation,
\begin{align}\label{eq:Deconvolution_binom}
\sum_{S\subseteq A} \pi(\thetav,S)\rho(\thetav,S) J_1\big(s-\log\big(\rho(\thetav,S)\big), \thetav\big) = J_0(s),\quad s\in\Rb, \thetav\in\Xi,
\end{align}
which is the convolution equation \eqref{eq:Deconvolution} in the binomial setting. Proposition \ref{prop:ExistenceUniquness} then yields the following result regarding the solution of \eqref{eq:Deconvolution_binom}.

\begin{corollary}
Assume that there exist constants $0<\gam_1\le\gam_2$ such that
\begin{align}\label{eq:Fmu_binom_condition}
\muh_{\thetav,k}(\xi):=\sum_{S\subseteq A} 
\big(\rho(\thetav,S)\big)^{1-\frac{1}{\gam_k}}
\pi(\thetav,S)
\ee^{\ii\log\big(\rho(\thetav,S)\big)\xi}\ne 0,\quad \xi\in\Rb,k\in\{1,2\},
\end{align}
and that $J_0\in\Jc(\gam_1,\gam_2)$ (with $\Jc(\gam_1,\gam_2)$ as in \eqref{eq:Jgam1gam2}). Assume further that
\begin{align}\label{eq:J1k_binom}
J_{1,k}(\cdot,\thetav):=
\F^{-1}\left[\frac{\F[J_{0,k}]}{\muh_{\thetav,k}}\right]\in L^1_\text{loc}(\Rb);
\quad \thetav\in\Xi, k\in\{1,2\},
\end{align}
in which $J_{0,1}(t):=J_0(t)\ee^{\frac{1}{\gam_1}t}\1_{\{t<0\}}$ and $J_{0,2}(t):=J_0(t)\ee^{\frac{1}{\gam_2}t}\1_{\{t\ge0\}}$ for $t\in\Rb$. Then, a solution of the functional equation \eqref{eq:Deconvolution_binom} is given by
\begin{align}
J_1(t,\thetav):=\ee^{-\frac{1}{\gam_1}t}J_{1,1}(t,\thetav) + \ee^{-\frac{1}{\gam_2}t}J_{1,2}(t,\thetav);\quad
(t,\thetav)\in\Rb\times\Xi.
\end{align}
Furthermore, $J_1(\cdot,\thetav)\in\Jc(\gam_1,\gam_2)$, $\thetav\in\Xi$.\qed
\end{corollary}
\begin{proof}
We have already checked that Assumption \ref{assum:gam1gam2} holds for any choice of $0<\gam_1\le\gam_2$. By \eqref{eq:Fmu_binom} $\F[\mu_{\thetav,k}]=\muh_{\thetav,k}$ and, thus, \eqref{eq:Fmu_binom_condition} yields that $1/\F[\mu_{\thetav,k}]\in\Cc^\infty$. Therefore, Condition (i) of Proposition \ref{prop:ExistenceUniquness} is satisfied. Condition (ii) of Proposition \ref{prop:ExistenceUniquness} is equivalent to \eqref{eq:J1k_binom}. Condition (iii) is also satisfied since the measure $\mu_{\thetav,k}$ in \eqref{eq:mu_binom} has finite support. The results then directly follows from Proposition \ref{prop:ExistenceUniquness}.
\end{proof}

Next, let us consider the entirety of Problem \ref{prob:IntEq_binom} but restrict the solution to CMIM functions, as we did in Subsection \ref{sec:CMIM}. From Theorem \ref{thm:CMIM}, we obtain the following result.

\begin{corollary}\label{cor:CMIM_binom}
For arbitrary constants $0<\gam_1\le\gam_2$, assume that $I_0\in\CMIM(\gam_1,\gam_2)$ (see Definition \ref{def:CMIM}). In particular, 
$I_0(y)=\int_{\gam_1}^{\gam_2} y^{-1/\gam}\dd m_0(\gam)$, $y>0$, for a finite Borel measure $m_0$ such that $\supp(m_0)\subset(\gam_1,\gam_2)$. Then,
\begin{align}\label{eq:IntEq_CMIM_sol_binom}
I_1(y,\thetav) := \int_{\gam_1}^{\gam_2} y^{-\frac{1}{\gam}}\left(
\sum_{S\subseteq A} \pi(\thetav,S)\rho(\thetav,S)^{1-\frac{1}{\gam}}
\right)^{-1}\dd m_0(\gam);\quad (y,\thetav)\in\Rb_+\times\Xi,
\end{align}
is the unique solution of Problem \ref{prob:IntEq} satisfying $I_1(\cdot,\thetav)\in\CMIM(\gam_1,\gam_2)$, $\thetav\in\Xi$.\qed
\end{corollary}
\begin{proof}
From \eqref{eq:nu_binom}, we have
\begin{align}
&\int_{\Rb_+} \left(\rho^{-\frac{1}{\gam_1}}+\rho^{1-\frac{1}{\gam_1}}+\rho^{1-\frac{1}{\gam_2}}\right)\dd\nu_\thetav(\rho)=\\
&\sum_{S\subseteq A} \pi(\thetav,S)\left(
\rho(\thetav,S)^{-\frac{1}{\gam_1}}+\rho(\thetav,S)^{1-\frac{1}{\gam_1}}
+\rho(\thetav,S)^{1-\frac{1}{\gam_2}}
\right)<+\infty.
\end{align}
Therefore, \eqref{eq:CMIM_nu} is satisfied and the corollary directly follows from Theorem \ref{thm:CMIM}.
\end{proof}

Finally, we state the following result which provides an explicit construction for PFPPs having completely monotonic inverse marginals in the binomial market. It directly follows by applying Theorem \ref{thm:PFPP_CMIM} to the generalized Binomial market, thus, we omit its proof.

\begin{corollary}\label{cor:PFPP_CMIM_binom}
For arbitrary constants $0<\gam_1\le\gam_2$, assume that $I_0\in\CMIM(\gam_1,\gam_2)$ (see Definition \ref{def:CMIM}). In particular, let
$I_0(y)=\int_{\gam_1}^{\gam_2} y^{-1/\gam}\dd m_0(\gam)$, $y>0$, for a finite Borel measure $m_0$ such that $\supp(m_0)\subset(\gam_1,\gam_2)$. Define the finite Borel measures $m_{n,\gv}$, $n\in\Nb$, $\gv\in\Gc_n$, by the iteration
\begin{align}\label{eq:m_n_CMIM_binom}
m_{n,\gv}(B) = \int_B\left(
\sum_{S\subseteq A_n} \pi_n(\thetav,S)\rho_n(\thetav,S)^{1-\frac{1}{\gam}}
\right)^{-1} \dd m_{n-1,\gv'}(\gam);\quad
B\in\Bc(\Rb), n\in\Nb, \gv\in\Gc_n,
\end{align}
in which $\gv=\gv'\oplus\thetav\in\Gc_n$ such that $\gv'\in\Gc_{n-1}$ and $\thetav\in\Xi_n$,
and let
\begin{align}\label{eq:In_CMIM_binom}
I_n(y,\gv) := \int_{\gam_1}^{\gam_2} y^{-\frac{1}{\gam}}\dd m_{n,\gv}(\gam);\quad n\in\Nb, \gv\in\Gc_n.
\end{align}
Then, the unique PFPP $\left\{U_n\right\}_{n\in\Nb_0}$ satisfying $I_0:=U_0^{\prime-1}$ and $U_n^{\prime-1}(\cdot,\gv)\in\CMIM(\gam_1,\gam_2)$, $n\in\Nb, \gv\in\Gc_n$, is given by
\begin{align}\label{eq:U_CMIM_binom}
U_n(x,\gv) &:= U_{n-1}\left(I_{n-1}(1,\gv'),\gv'\right) +\Eb\left[\int_{I_n\left(\rho_n, \gv\right)}^x I_n^{-1}(\xi,\gv) \dd\xi\middle|\Gv_n=\gv\right];\quad
x\in\Rb_+, n\in\Nb, \gv\in\Gc_n,
\end{align}
in which $\gv=\gv'\oplus\thetav\in\Gc_n$ (such that $\gv'\in\Gc_{n-1}$ and $\thetav\in\Xi_n$). 
A corresponding optimal wealth process starting with initial wealth $x_0>0$ is given by
\begin{align}\label{eq:XStar_CMIM_binom}
X^*_n:=I_n\left(\rho_n I_{n-1}^{-1}(X^*_{n-1},\Gv_{n-1}), \Gv_n\right);\quad n\in\Nb,
\end{align}
with $X^*_0=x_0$.\qed
\end{corollary}

%
%
\subsection{PFPPs in a generalized Black-Scholes model}\label{sec:BS}
In our second example, we consider the generalized Black-Scholes market of Example \ref{ex:BS}. By Remark \ref{rem:Binom_BS_rho}, we have that $\Thetav_n=\Lamv_n$ and $\Gv_n=(\Lamv_1,\dots,\Lamv_n)$. Recall that $\Lamv_n$ is the vector of the Sharpe ratios of the risky assets for period $[n-1,n]$, see \eqref{eq:S_BS}. For $n\in\Nb$, we have $\Xi_n=\Rb^K$ and $\Gc_n=\big(\Rb^K\big)^n$, since $\Lamv_n$ is assumed to be a $\Rb^K$-valued random vector (recall that $K\ge1$ is the number of risky assets).

From Remark \ref{rem:Binom_BS_rho}, we have
\begin{align}\label{eq:rho_n_BS}
\rho_n = \exp\left(-\frac{1}{2}\|\Lamv_n\|^2 - \Lamv_n^\top(\Bv_n-\Bv_{n-1})\right),\quad n\in\Nb,
\end{align}
in which $\Bv=(\Bv_t)_{t\ge0}$ is the $K$-dimensional standard Brownian motion. In Example \ref{ex:BS}, we assumed that $(\Bv_t-\Bv_n)_{t\ge n}$ is independent of $\{(\Lamv_m,\Sig_m)\}_{m=1}^{n+1}$ for all $n\in\Nb_0$. Therefore, $\rho_n|_{\Gv_n}=\rho_n|_{\Lamv_n}$ (i.e. the conditional law of $\rho_n$ given $\Gv_n=(\Lamv_1,\dots,\Lamv_n)$ is the same as the conditional law of $\rho_n$ given $\Lamv_n$). By \eqref{eq:rho_n_BS}, $\rho_n|_{\Lamv_n=\lamv}$, $(n,\lamv)\in\Nb\times\Rb^K$, has log-normal distribution, that is,
\begin{align}
\Pb\left(\rho_n\in B \middle|\Lamv_n=\lamv\right)
=\int_B \frac{1}{\rho\sqrt{2\pi\|\lamv\|^2}}\exp\left(-\frac{\left(\log \rho + \frac{1}{2}\|\lamv\|^2\right)^2}{2\|\lamv\|^2}\right)\dd \rho,\quad
B\in\Bc(\Rb_+), \lamv\in\Rb^K.
\end{align}
Note that $\supp(\rho_n|_{\Lamv_n=\lamv})=\Rb_+$ and is not a compact set.

Using these notations, Problem \eqref{prob:IntEq} takes the following form in the Black-Scholes model.
\begin{problem}\label{prob:IntEq_binom_BS}
Given an $I_0\in\Ic$, find an $I_1:\Rb_+\times\Rb^K\to\Rb_+$ such that	\begin{align}\label{eq:IntEq_binom_BS}
\int_{\Rb_+} \rho I_1(y\rho, \lamv)\dd \nu_\lamv(\rho) = I_0(y),
\end{align}
$\int_{\Rb_+} I_1(y\rho, \lamv)\dd \nu_\lamv(\rho)<\infty$, and $I_1(\cdot,\lamv)\in\Ic$, for all $y>0$ and $\lamv\in\Rb^K$, in which
\begin{align}\label{eq:nu_BS}
d\nu_\lamv(\rho) = \frac{1}{\rho\sqrt{2\pi\|\lamv\|^2}}\exp\left(-\frac{\left(\log \rho + \frac{1}{2}\|\lamv\|^2\right)^2}{2\|\lamv\|^2}\right)\dd \rho.
\end{align}\qed
\end{problem}

Let us first consider the integral equation \eqref{eq:IntEq_binom_BS}. Following our discussion in Section \ref{sec:deconvolution}, set $y=\ee^s$, $\rho=\ee^{-t}$, $J_0(t):= I_0(\ee^t)$, and $J_1(t,\lamv) := I_1(\ee^t, \lamv)$ to transform \eqref{eq:IntEq_binom_BS} into the following convolution equation
\begin{align}\label{eq:Deconvolution_BS}
\int_\Rb J_1(s-t,\lamv)\dd\nut_\lamv(t) = J_0(s),\quad s\in\Rb, \lamv\in\Rb^K,
\end{align}
in which the measure $\nut_\lamv$ (which corresponds to $\nut_\thetav$ of \eqref{eq:nut}) is given by
\begin{align}\label{eq:nut_BS}
\dd\nut_\lamv(t) = \frac{1}{\sqrt{2\pi\|\lamv\|^2}}\exp\left(-\frac{\left(t+\frac{1}{2}\|\lamv\|^2\right)^2}{2\|\lamv\|^2}\right)\dd t,
\end{align}
i.e. $\nut_\lamv$ is a Gaussian probability measure with mean $-\|\lamv\|^2/2$ and variance $\|\lamv\|^2$.

Next, we check that Assumption \eqref{assum:gam1gam2} is satisfied. Take arbitrary constants $0<\gam_1\le\gam_2$. From \eqref{eq:mu_thetav}, we define the measures $\mu_{\lamv,k}$, $(\lamv,k)\in\Rb^K\times\{1,2\}$, as follows
\begin{align}\label{eq:mu_lam_k_BS}
\dd \mu_{\lamv,k}(t) := \ee^{\frac{t}{\gam_k}}\dd\nut_\lamv(t)
=\ee^{\frac{1}{2}\|\lamv\|^2\frac{1}{\gam_k}\left(\frac{1}{\gam_k}-1\right)\,} \Phi_{\lamv,k}(t)\dd t,\quad t\in\Rb, k\in\{1,2\},
\end{align} 
in which $\Phi_{\lamv,k}(t)$ is a Gaussian density function with mean $\|\lamv\|^2\left(\frac{1}{\gam_k}-\frac{1}{2}\right)$ and variance $\|\lamv\|^2$. For $(\lamv,k)\in\Rb^K\times\{1,2\}$, we have $\mu_{\lamv,k}\in\Sc'$ as it is the product of the constant $\ee^{\frac{1}{2}\|\lamv\|^2\frac{1}{\gam_k}\left(\frac{1}{\gam_k}-1\right)\,}$ and a probability distribution. Furthermore, $\F[\mu_{\lamv,k}]\in\Cc^\infty$ since
\begin{align}
\F[\mu_{\lamv,k}](\xi)
= \ee^{\frac{1}{2}\|\lamv\|^2\frac{1}{\gam_k}\left(\frac{1}{\gam_k}-1\right)\,}\F[\Phi_{\lamv,k}](\xi)
= \ee^{-\frac{1}{2}\|\lamv\|^2\left[
\xi^2 + 2\ii \xi \left(\frac{1}{\gam_k}-1\right)
-\frac{1}{\gam_k}\left(\frac{1}{\gam_k}-1\right)
\right]};\quad \xi\in\Rb.
\end{align}
We have shown that Assumption \eqref{assum:gam1gam2} is satisfied in the Black-Scholes market for any choice of $0<\gam_1\le\gam_2$.

We then obtain the following corollary of Proposition \ref{prop:ExistenceUniquness} regarding the solution of \eqref{eq:Deconvolution_BS}.

\begin{corollary}\label{cor:ExistenceUniquness_BS}
Assume that $J_0\in\Jc(\gam_1,\gam_2)$ for some constants $0<\gam_1\le\gam_2$ and with $\Jc(\gam_1,\gam_2)$ as in \eqref{eq:Jgam1gam2}, and define $J_{0,1}(t):=J_0(t)\ee^{\frac{1}{\gam_1}t}\1_{\{t<0\}}$ and $J_{0,2}(t):=J_0(t)\ee^{\frac{1}{\gam_2}t}\1_{\{t\ge0\}}$ for $t\in\Rb$. Assume further that for all $\lamv\in\Rb^K$ and $k\in\{1,2\}$:\vspace{1ex}\\
\noindent$(i)$ $\ee^{\frac{1}{2}\|\lamv\|^2\left[
\xi^2 + 2\ii \xi \left(\frac{1}{\gam_k}-1\right)
-\frac{1}{\gam_k}\left(\frac{1}{\gam_k}-1\right)
\right]}\F[J_{0,k}]\in\Sc'$,\\
\noindent$(ii)$ $	J_{1,k}(\cdot,\lamv):=
\F^{-1}\left[
\ee^{\frac{1}{2}\|\lamv\|^2\left[
\xi^2 + 2\ii \xi \left(\frac{1}{\gam_k}-1\right)
-\frac{1}{\gam_k}\left(\frac{1}{\gam_k}-1\right)
\right]}
\F[J_{0,k}]
\right]\in L^1_\text{loc}(\Rb)$, and\\
\noindent$(iii)$ $\int_{\Rb} |J_{1,k}(s-t,\lamv)|\dd\mu_{\lamv,k}(t)<+\infty$ for all $s\in\Rb$ and with $\mu_{\lamv,k}$ given by \eqref{eq:mu_lam_k_BS}.\vspace{1ex}\\
Then, a solution of the convolutions equation \eqref{eq:Deconvolution_BS} is given by
\begin{align}
J_1(t,\lamv):=\ee^{-\frac{1}{\gam_1}t}J_{1,1}(t,\lamv) + \ee^{-\frac{1}{\gam_2}t}J_{1,2}(t,\lamv);\quad
(t,\lamv)\in\Rb\times\Rb^K.
\end{align}
Furthermore, $J_1(\cdot,\lamv)$.\qed
\end{corollary}
\begin{proof}
As we have already confirmed, Assumption \ref{assum:gam1gam2} holds for any choice of $0<\gam_1\le\gam_2$. In light of \eqref{eq:mu_lam_k_BS}, the corollary then directly follows from Proposition \ref{prop:ExistenceUniquness}.
\end{proof}

Next, we apply the analysis of Subsection \ref{sec:CMIM} to solve Problem \ref{prob:IntEq_binom} while restricting ourselves to solution with inverse marginals that are completely monotonic. From Theorem \ref{thm:CMIM}, we obtain the following result.

\begin{corollary}\label{cor:CMIM_BS}
For arbitrary constants $0<\gam_1\le\gam_2$, assume that $I_0\in\CMIM(\gam_1,\gam_2)$ (see Definition \ref{def:CMIM}). In particular, let
$I_0(y)=\int_{\gam_1}^{\gam_2} y^{-1/\gam}\dd m_0(\gam)$, $y>0$, for a finite Borel measure $m_0$ such that $\supp(m_0)\subset(\gam_1,\gam_2)$. Then,
\begin{align}\label{eq:IntEq_CMIM_sol_BS}
I_1(y,\lamv) := \int_{\gam_1}^{\gam_2} y^{-\frac{1}{\gam}}
\exp\left(\frac{\|\lamv\|^2}{2\gam}\left(1-\frac{1}{\gam}\right)\right)
\dd m_0(\gam);\quad (y,\lamv)\in\Rb_+\times\Rb^K,
\end{align}
is the unique solution of Problem \ref{prob:IntEq} satisfying $I_1(\cdot,\lamv)\in\CMIM(\gam_1,\gam_2)$, $\lamv\in\Rb^K$.\qed
\end{corollary}
\begin{proof}
It follows from \eqref{eq:nu_BS} that $\int_{\Rb_+}\rho^a \dd \nu_\lamv(\rho) = \ee^{\frac{1}{2}\|\lamv\|^2 a(a-1)}$ for all $a\in\Rb$ and $\lamv\in\Rb^K$.	Therefore, \eqref{eq:CMIM_nu} is satisfied and the corollary directly follows from Theorem \ref{thm:CMIM}.
\end{proof}

We end this section by the following corollary of Theorem \ref{thm:PFPP_CMIM}, which provides an explicit procedure for constructing PFPPs with completely monotonic inverse marginal functions. Its proof directly follows from Theorem \ref{thm:PFPP_CMIM} and is thus omitted.

\begin{corollary}\label{cor:PFPP_CMIM_BS}
For arbitrary constants $0<\gam_1\le\gam_2$, assume that $I_0\in\CMIM(\gam_1,\gam_2)$ (see Definition \ref{def:CMIM}) such that
$I_0(y)=\int_{\gam_1}^{\gam_2} y^{-1/\gam}\dd m_0(\gam)$, $y>0$, for a finite Borel measure $m_0$ with $\supp(m_0)\subset(\gam_1,\gam_2)$. Define the finite Borel measures $m_{n,\gv}$, $n\in\Nb$, $\gv\in(\Rb^K)^n$, through the iteration
\begin{align}\label{eq:m_n_CMIM_BS1}
m_{1,\lamv}(B) = \int_B \exp\left(\frac{\|\lamv\|^2}{2\gam}\left(1-\frac{1}{\gam}\right)\right)
\dd m_0(\gam);\quad
B\in\Bc(\Rb), \lamv\in\Rb^K,
\end{align}
and
\begin{align}\label{eq:m_n_CMIM_BS}
m_{n,(\lamv_1,\dots,\lamv_n)}(B) = \int_B \exp\left(\frac{\|\lamv_n\|^2}{2\gam}\left(1-\frac{1}{\gam}\right)\right)
\dd m_{n-1,(\lamv_1,\dots,\lamv_{n-1})}(\gam),
\end{align}
for $B\in\Bc(\Rb)$, $n\ge2$, and $\lamv_1,\dots,\lamv_n\in\Rb^K$. Note that $\supp(m_{n,\gv})=\supp(m_0)\subset(\gam_1,\gam_2)$, $n\in\Nb$, $\gv\in(\Rb^K)^n$. 
Let
\begin{align}\label{eq:In_CMIM_BS}
I_n(y,\gv) := \int_{\gam_1}^{\gam_2} y^{-\frac{1}{\gam}}\dd m_{n,\gv}(\gam);\quad n\in\Nb, \gv\in(\Rb^K)^n.
\end{align}
Then, the unique PFPP $\left\{U_n\right\}_{n\in\Nb_0}$ satisfying $I_0:=U_0^{\prime-1}$ and $U_n^{\prime-1}(\cdot,\gv)\in\CMIM(\gam_1,\gam_2)$, $n\in\Nb, \gv\in(\Rb^K)^n$, is given by
\begin{align}\label{eq:U_CMIM_BS}
U_n\big(x,\gv) &:= U_{n-1}\left(I_{n-1}(1,\gv'),\gv'\right) +\Eb\left[\int_{I_n\left(\rho_n, \gv\right)}^x I_n^{-1}(\xi,\gv) \dd\xi\middle|(\Lamv_1,\dots,\Lamv_n)=\gv\right],
\end{align}
for $x>0$, $n\in\Nb$, $\gv=(\lamv_1,\dots,\lamv_n)$, $\gv'=(\lamv_1,\dots,\lamv_{n-1})$, and $\lamv_1,\dots,\lamv_n\in\Rb^K$.
A corresponding optimal wealth process starting with initial wealth $x_0>0$ is given by
\begin{align}\label{eq:XStar_CMIM_BS}
X^*_n:=I_n\left(\rho_n I_{n-1}^{-1}\big(X^*_{n-1},(\Lamv_1,\dots,\Lamv_{n-1})\big), (\Lamv_1,\dots,\Lamv_n)\right);\quad n\in\Nb,
\end{align}
with $X^*_0=x_0$.\qed
\end{corollary}

%
%


%
%
\appendix

%
%
\section{Proof of Theorem \ref{thm:verif}}\label{app:verif}
We start with deriving a few auxiliary results, and then check that $\left\{U_n\right\}_{n\in\Nb_0}$ satisfy Conditions $(i)$--$(iii)$ of Definition \ref{def:PFPP} and thus is a PFPP.

Choose arbitrarily an $n\in\Nb$ and a $\gv=\gv'\oplus\thetav\in\Gc_n$ such that $\gv'\in\Gc_{n-1}$ and $\thetav\in\Xi_n$. 
The first part of Condition $(i)$ yields that $I_n(\cdot,\gv)$ and $I_n^{-1}(\cdot,\gv)$ are decreasing. Therefore, for any $y>0$,
\begin{align}
	&\left|\int_{I_n\left(y\rho_n, \gv\right)}^{I_n\left(1, \gv\right)} I_n^{-1}(\xi,\gv) \dd\xi\right|
	\le
	\max\Big\{I_n^{-1}\big(I_n(y\rho_n, \gv),\gv\big), I_n^{-1}\big(I_n(1, \gv),\gv\big)\Big\}\,
	\big|I_n(y\rho_n, \gv)- I_n(1, \gv)\big|\\
	&\le\max\{y\rho_n,1\}\big|I_n(y\rho_n, \gv)- I_n(1, \gv)\big|
	\le (1+y\rho_n) \big(I_n(y\rho_n, \gv)+ I_n(1, \gv)\big).
\end{align}
By taking conditional expectation given $\Gv_n=\gv$, it follows that
\begin{align}\label{eq:Condition_iii}
	\Eb\left[\left|\int_{I_n\left(y\rho_n, \gv\right)}^{I_n\left(1, \gv\right)} I_n^{-1}(\xi,\gv) \dd\xi\right|~\middle|\Gv_n=\gv\right]
	\le \Eb\Big[(1+y\rho_n) \big(I_n(y\rho_n, \gv)+ I_n(1, \gv)\big)~\Big|\Gv_n=\gv\Big]
	<+\infty,
\end{align}
in which we have used the second part of Condition $(i)$ for the second inequality.
By \eqref{eq:U}, we have that
\begin{align}\label{eq:U_alt}
	U_n(x,\gv) = &\int_{I_n\left(1, \gv\right)}^x I_n^{-1}(\xi,\gv) \dd\xi\\
	&+ U_{n-1}\left(I_{n-1}(1,\gv'),\gv'\right)
	+ \Eb\left[\int_{I_n\left(\rho_n, \gv\right)}^{I_n\left(1, \gv\right)} I_n^{-1}(\xi,\gv) \dd\xi\middle|\Gv_n=\gv\right].
\end{align}
By \eqref{eq:Condition_iii} and induction on $n\in\Nb$, it follows that $U_n$ is Borel measurable and that $U_n(\cdot,\gv)\in \Cc^2(\Rb_+)$ for all $\gv\in\Gc_n$. In particular,
\begin{align}\label{eq:Un_In}
	U'_n(x,\gv) := \frac{\partial}{\partial x}U_n(x,\gv) = \frac{\partial}{\partial x}\left(\int_1^x I_n^{-1}(\xi,\gv) \dd\xi\right)= I_n^{-1}(x,\gv).
\end{align}
By setting $y=U'_{n-1}(x,\gv')$ in Condition $(ii)$ and then using \eqref{eq:Un_In}, we obtain that
\begin{align}\label{eq:IntEq_alt}
	\Eb\left[\rho_n I_n\left(U_{n-1}'(x,\gv')\rho_n, \gv\right)\middle|\Gv_n=\gv\right] = x,
\end{align}
for all $(n,y,\gv=\gv'\oplus\thetav)\in\Nb\times\Rb_+\times\Gc_n$ (such that $\gv'\in\Gc_{n-1}$ and $\thetav\in\Xi_n$).

Define
$\Ut(x) :=\Eb\left[U_n\Big(I_n\big(U'_{n-1}(x,\gv')\rho_n,\gv\big),\gv\Big)\middle|\Gv_n=\gv\right]$, $x>0$. We will show that $\Ut(x)=U_{n-1}(x,\gv')$, $x>0$.	From \eqref{eq:U_alt}, it follows that
\begin{align}
	\Ut(x)
	= &\Eb\left[\int_{I_n\left(1, \gv\right)}^{I_n\big(U'_{n-1}(x,\gv')\rho_n,\gv\big)} I_n^{-1}(\xi,\gv) \dd\xi\middle|\Gv_n=\gv\right]
	+ U_{n-1}\left(I_{n-1}(1,\gv'),\gv'\right)\\
	\label{eq:Ut}
	&{}+ \Eb\left[\int_{I_n\left(\rho_n, \gv\right)}^{I_n\left(1, \gv\right)} I_n^{-1}(\xi,\gv) \dd\xi\middle|\Gv_n=\gv\right],
\end{align}
and, thus, $\Ut$ has a continuous derivative because of \eqref{eq:Condition_iii}.	
In particular,
\begin{align}
	\Ut'(x)&=\Eb\left[U''_{n-1}(x,\gv')\rho_n I_n'\big(U'_{n-1}(x,\gv')\rho_n,\gv\big)U_n'\Big(I_n\big(U'_{n-1}(x,\gv')\rho_n,\gv\big),\gv\Big)\middle|\Gv_n=\gv\right]\\
	\label{eq:Utp0}
	&=U'_{n-1}(x,\gv')\Eb\left[\rho_n^2 U''_{n-1}(x,\gv') I_n'\big(U'_{n-1}(x,\gv')\rho_n,\gv\big)\middle|\Gv_n=\gv\right].
\end{align}
Differentiating \eqref{eq:IntEq_alt} with respect to $x$ yields that
\begin{align}
	\Eb\left[\rho_n^2 U_{n-1}''(x,\gv')I_n'\left(U_{n-1}'(x,\gv')\rho_n, \gv\right)\middle|\Gv_n=\gv\right] = 1.
\end{align}
From \eqref{eq:Utp0}, it then follows that $\Ut'=U_{n-1}'(\cdot,\gv')$. Finally, setting $x=I_{n-1}(1,\gv')$ in \eqref{eq:Ut} yields that $\Ut\big(I_{n-1}(1,\gv')\big)=U_{n-1}\left(I_{n-1}(1,\gv'),\gv'\right)$, 
and we must have $\Ut(x)=U_{n-1}(x,\gv)$, $x>0$, as we set out to prove. We have shown that
\begin{align}\label{eq:Uold_Unew}
	U_{n-1}(x,\gv')=\Eb\left[U_n\Big(I_n\big(U'_{n-1}(x,\gv')\rho_n,\gv\big),\gv\Big)\middle|\Gv_n=\gv\right],
\end{align}
for any $(n,x,\gv=\gv'\oplus\thetav)\in\Nb\times\Rb_+\times\Gc_n$ such that $\gv'\in\Gc_{n-1}$ and $\thetav\in\Xi_n$.\footnote{For $n=1$, \eqref{eq:Uold_Unew} becomes $U_0(x)=\Eb\left[U_1\Big(I_1\big(U'_0(x)\rho_1,\thetav\big),\thetav\Big)\middle|\Thetav_1=\thetav\right]$ for all $(x,\thetav)\in\Rb_+\times\Xi_1$.}

We are now ready to check that $\left\{U_n\right\}_{n\in\Nb_0}$ satisfy Conditions $(i)$--$(iii)$ of Definition \ref{def:PFPP}.
\vspace{1ex}

\noindent\textbf{Condition (i) of Definition \ref{def:PFPP}:} 
The first part of Condition $(i)$ and \eqref{eq:Un_In} yields that $U_n(\cdot,\gv)\in\Uc$ for $(n,\gv)\in\Nb\times\Gc_n$. Furthermore, $U_0\in\Uc$ by assumption.\vspace{1ex}

\noindent\textbf{Condition (ii) of Definition \ref{def:PFPP}:} Take an arbitrary choice for $(n,x,\gv=\gv'\oplus\thetav)\in\Nb\times\Rb_+\times\Gc_n$ (such that $\gv'\in\Gc_{n-1}$ and $\thetav\in\Xi_n$) and assume that $X$ is any random variable satisfying $X\in\Ac_n(x)$ and $\Eb\big[U_n(X,\gv)\big|\Gv_n=\gv\big]>-\infty$.	
Let $V_n(\cdot,\gv):\Rb_+\to\Rb$ be the convex dual of the utility function $U_n(\cdot,\gv)$, namely,
\begin{align}\label{eq:Vn}
	V_n(y,\gv) := \sup_{x>0}\left\{U_n(x,\gv) - xy\right\} = U_n\big(I_n(y,\gv),\gv\big)-yI_n(y,\gv);\quad y>0,
\end{align}
in which the second equality follows from \eqref{eq:Un_In}.
By Lemma \ref{lem:ConvexDuality}, we have that $V_n(\cdot,\gv)\in\Cc^2(\Rb_+)$ (in particular, $V'_n(y,\gv):=\frac{\partial}{\partial y}V_n(y,\gv) = -I_n(y,\gv)<0$ and $V''_n(y,\gv):=\frac{\partial^2}{\partial y^2}V_n(y,\gv) = -1/U_n''\big(I_n(y,\gv),\gv\big)>0$ for all $y>0$).	
By \eqref{eq:IntEq_alt} and since $X\in\Ac_n(x)$, we have
\begin{align}
	\Eb\left[\rho_n I_n\big(U'_{n-1}(x,\gv')\rho_n, \gv\big)\middle|\Gv_n=\gv\right] = x
	= \Eb\left[\rho_n X \middle|\Gv_n=\gv\right].
\end{align}
From this equation and \eqref{eq:Vn}, it follows that
\begin{align}
	&\Eb\left[ U'_{n-1}(x,\gv') \rho_n X + V_n\big(U'_{n-1}(x,\gv')\rho_n, \gv\big)\middle|\Gv_n=\gv\right]\\
	&= \Eb\Big[\rho_n U'_{n-1}(x,\gv') I_n\big(U'_{n-1}(x,\gv')\rho_n, \gv\big)
	+ V_n\big(U'_{n-1}(x,\gv')\rho_n, \gv\big)\Big|\Gv_n=\gv\Big]\\
	\label{eq:FinitExp}
	&= \Eb\left[U_n\Big(I_n\big(U'_{n-1}(x,\gv')\rho_n,\gv\big),\gv\Big)\middle|\Gv_n=\gv\right]
	=U_{n-1}(x,\gv'),
\end{align}
in which that last step follows from \eqref{eq:Uold_Unew}.
From \eqref{eq:Vn}, we have that
$
U_n(X,\gv)\le U_{n-1}'(x,\gv')\rho_n X + V_n\big(U'_{n-1}(x,\gv')\rho_n, \gv\big).
$ 
By taking expectation and using \eqref{eq:FinitExp}, we then obtain that
\begin{align}
	\Eb\Big[U_n(X,\gv)\Big|\Gv_n=\gv\Big]
	\le \Eb\left[U_n\Big(I_n\big(U'_{n-1}(x,\gv')\rho_n,\gv\big),\gv\Big)\middle|\Gv_n=\gv\right]
	=U_{n-1}(x,\gv').
\end{align}
Thus, Condition (ii) of Definition \ref{def:PFPP} is satisfied.\vspace{1em}	

\noindent\textbf{Condition (iii) of Definition \ref{def:PFPP}:}
Let $\{X^*_n\}_{n\in\Nb_0}$ be as defined in the statement of the Theorem. For $n\in\Nb$, we have that $X^*_n\in\Ac_n(X^*_{n-1})$ since
\begin{align}
	\Eb[Z_n X^*_n |\Fc_{n-1}]
	&= Z_{n-1}\Eb\left[\rho_n I_n\big(\rho_n I_{n-1}^{-1}(X^*_{n-1},\Gv_{n-1}), \Gv_n\big)\middle|\Fc_{n-1}\right]\\\label{eq:AssumpRho_needed}
	&= Z_{n-1}\Eb\left[\rho_n I_n\big(\rho_n I_{n-1}^{-1}(X^*_{n-1},\Gv_{n-1}), \Gv_n\big)\middle|\Gv_n\right] = Z_{n-1} X^*_{n-1},
\end{align}
in which we used Assumption \ref{asum:rho} for the second step and  \eqref{eq:IntEq_alt} for the last step. Thus, $\{X^*_n\}_{n\in\Nb_0}\in\Ac$ by \eqref{eq:Act}. From \eqref{eq:Uold_Unew}, it follows that $U_{n-1}(x,\gv') = \Eb\big[U_n(X^*_n,\gv)\big|X^*_{n-1}=x,\Gv_n=\gv\big]$ for all $(n,x,\gv=\gv'\oplus\thetav)\in\Nb\times\Rb_+\times\Gc_n$ such that $\gv'\in\Gc_{n-1}$ and $\thetav\in\Xi_n$. Lemma \ref{lem:DisceretWealth} then yields that $\left(X_t^*:=\Eb\left[X^*_{\ceil{t}}|\Fc_t\right]\right)_{t\ge0}\in\Act$ satisfies Condition $(iii)$ of Definition \ref{def:PFPP}.

%
%

\section{Review of the Fourier transform for tempered distributions}\label{sec:FourierAnalysis}
This appendix provides a brief summary of the results used in Section \ref{sec:deconvolution} from the theory of distributions and the Fourier transform for tempered distributions. For more details and proofs, we refer the reader to any modern treatment of the Fourier analysis, for instance, \cite{Hormander1990}.

%
%

\subsection{Distributions}\label{sec:Dist}

Let $\Cb$ be the set of complex numbers. For an open subset $D\subseteq\Rb$, let $\Cc^{\infty}(D)$ denote the set of all $\Cb$-valued infinitely-differentiable functions with domain $D$, and $\Cc^{\infty}_0(D)$ be the set of all elements of $\Cc^{\infty}(D)$ with compact support. We take the convention that $\Cc^{\infty}=\Cc^{\infty}(\Rb)$ and $\Cc_0^{\infty}=\Cc_0^{\infty}(\Rb)$. Endow $\Cc^{\infty}$ and $\Cc_0^{\infty}$ with the topology generated by the family of seminorms
\begin{align}\label{eq:C_norm}
	\|\varphi\|_{n,\chi} := \sum_{k\in\{0,\dots,n\}} \sup_{x\in\chi}\left|\frac{\dd^k \varphi(x)}{\dd x^k}\right|,
\end{align}
in which $n\in\Nb_0$ and $\chi$ is a compact subset of $\Rb$. A \emph{distribution} (also called a \emph{generalized function}, and not to be confused with a probability distribution) is a continuous linear functional $f:\Cc^\infty_0\to \Cb$, and $\Dc'$ denotes the set of all distributions.

\begin{definition}\label{def:dist}
	The space of distributions $\Dc'(D)$ is the dual space of $\Cc_0^\infty(D)$ with the topology generated by the seminorms $\|\cdot\|_{n,\chi}$ of \eqref{eq:C_norm}. In other words, a linear functional $f:\Cc_0^\infty(D)\to\Cb$ is a distribution if for any compact set $\chi\subset D$, there exists an integer $n_\chi\in\Nb_0$ and a constant $C_\chi>0$ such that $\big|f(\varphi)\big| \le C_\chi \|\varphi\|_{n_\chi,\chi}$ for all $\varphi\in\Cc^\infty_0(D)$. We take the convention that $\Dc'=\Dc'(\Rb)$.\qed
\end{definition}

Any continuous function $\ft:\Rb\to\Rb$ is represented by a distributions $f\in\Dc'$ given by
\begin{align}\label{eq:fn_dist}
	f(\varphi):=\int_\Rb \ft(x)\varphi(x)\dd x;\quad \varphi\in\Cc^{\infty}_0.
\end{align}
This representation is unique in the sense that if $\gt\in\Cc^0_0$ also satisfies $f(\varphi):=\int_\Rb \gt(x)\varphi(x)\dd x$, $\varphi\in\Cc^{\infty}_0$, then we must have $\gt(x)=\ft(x)$ for $x\in\Rb$ (Theorem 1.2.4 on page 15 of \cite{Hormander1990}). Indeed, the representation \eqref{eq:fn_dist} holds for any $\ft\in L^1_\text{loc}$, that is, all functions $\ft:\Rb\to\Rb$ that are (Lebesgue) integrable on compact subsets of $\Rb$. For such $\ft$, the representation \eqref{eq:fn_dist} is unique almost everywhere on $\Rb$. In a similar fashion, a $\sig$-finite measure $\mut$ on $\Rb$ can be identified as a distribution $\mu\in\Dc'$ defined by
\begin{align}\label{eq:measure_dist}
	\mu(\varphi):=\int_\Rb \varphi(x)\mut(\dd x);\quad \varphi\in\Cc^{\infty}_0.
\end{align}
A simple example of such distribution is the \emph{Dirac measure} $\delta_a$, $a\in\Rb$, given by
\begin{align}\label{eq:DeltaFn}
	\delta_a(\varphi) := \varphi(a);\quad \varphi\in\Cc^{\infty},
\end{align}
which corresponds to a probability measure with a single atom at $a$ and mass $1$. 
As is customary in the literature and with a slight abuse of notation, we do not distinguish between a locally integrable function $\ft$ (respectively, a $\sig$-finite measure $\mut$) and the corresponding distribution $f$ (respectively, $\mu$) given by \eqref{eq:fn_dist} (respectively, \eqref{eq:measure_dist}). With this convention, we consider locally integrable functions and $\sig$-finite measures as distributions.

Next, we define the support of a distribution. Let $f\in\Dc'$ and $D$ be an open subset of $\Rb$. Then the \emph{restriction of $f$ to $D$} is the distribution $f_D\in\Dc'(D)$ given by $f_D(\varphi):=f(\varphi)$, $\varphi\in\Cc_0^\infty(D)$. If $f\in\Dc'(D)$ and for every $a\in D$ there exists an open set $N_a\subset D$ containing $a$ such that $f_{N_a}=0$, then $f=0$ (Theorem 2.2.1 on page 41 \cite{Hormander1990}). We define the \emph{support} of a distribution as follows.
\begin{definition}\label{def:Support}
	Let $D\subseteq\Rb$ be an open set and $f\in\Dc'(D)$. Let $N$ be the set of all points $a\in D$ such that $f_{N_a}=0$ for an open set $N_a\subset D$ containing $a$. We define $\supp(f):=N^c=D\backslash N$. Note that $N=\supp(f)^c$ is an open set and that $f_{\supp(f)^c}=0$.\qed
\end{definition}

An important class of distributions is the space of \emph{distributions with compact support}, denoted by $\mathcal{E}^\prime$. It can be defined in two forms. See Theorem 2.3.1 on page 44 of \cite{Hormander1990} for the equivalence of the definitions.

\begin{definition}\label{def:compdist}
	For an open set $D\subseteq\Rb$, $\Ec^\prime(D)$ is the set of all distributions $f\in\Dc'(D)$ such that $\supp(f)$ is compact. Equivalently, $\Ec'(D)$ is the dual space of $\Cc^\infty(D)$ with the topology generated by the seminorms $\|\cdot\|_{n,\chi}$ given by \eqref{eq:C_norm}. In other words, a linear functional $f:\Cc^\infty\to\Cb$ is a distribution with compact support if for any compact set $\chi\subset D$, there exists an integer $n_\chi\in\Nb_0$ and a constant $C_\chi>0$ such that $\big|f(\varphi)\big| \le C_\chi \|\varphi\|_{n_\chi,\chi}$ for all  $\varphi\in\Cc^\infty(D)$. Note that $\Ec'(D)\subset \Dc'(D)$ since $\Cc_0^{\infty}(D)\subset\Cc^{\infty}(D)$. We take the convention that $\Ec'=\Ec'(\Rb)$.\qed
\end{definition}

Next, we define two operations on distributions, namely, multiplying by smooth functions and differentiation. Let $f\in L^1_\text{loc}$ and $\psi\in\Cc^{\infty}$. We have that $\int_\Rb \big(\psi(x)f(x)\big)\varphi(x) \dd x=\int_\Rb f(x)\big(\psi(x)\varphi(x)\big) \dd x$ for all $\varphi\in\Cc_0^{\infty}$ and that $\psi\varphi\in\Cc_0^{\infty}$. Furthermore, for $f\in\Cc^1$, integration-by-parts yields that $\int_\Rb f'(x)\varphi(x)\dd x = -\int_\Rb f(x)\varphi'(x)\dd x$ for all $\varphi\in\Cc_0^{\infty}$. In light of correspondence between functions and distributions as in \eqref{eq:fn_dist}, these observations motivate the definition of (weak) differentiation and multiplication by smooth functions for a distribution.
\begin{definition}\label{def:DerivProdDist}
	Let $D$ be an open subset of $\Rb$ and $f\in\Dc'(D)$. Then, the distribution $f'\in\Dc'(D)$ is defined by $f'(\varphi) := -f(\varphi^\prime)$, $\varphi\in C^\infty_0(D)$.
	Furthermore, for $\psi\in C^\infty(D)$, the distribution $\psi f \in\Dc'(D)$ is defined by $\psi u(\varphi) := u(\psi\varphi)$, $\varphi\in C^\infty_0(D)$. We have that $\supp(u^\prime),\supp(\psi u)\subseteq\supp(u)$ and that the mappings $u\mapsto u^\prime$ and $u\mapsto \psi u$ are continuous on $\Dc'(D)$.\qed
\end{definition}

For $f\in\Dc'(D)$, we say that $f\ge0$ if $f(\varphi)\ge0$ for all $\varphi\in\Cc_0^{\infty}(D)$ such that $\varphi(x)\ge0$, $x\in D$. The following Lemma shows that distributions with non-negative first derivative are non-decreasing functions and those with non-negative second derivative are convex functions.

\begin{lemma}\label{lem:IncCovexDist}
	(Theorem 4.1.6 on page 90 of \cite{Hormander1990})
	Let $D$ be an open subset of $\Rb$ and $f\in\Dc'(D)$. Then, $f^\prime \ge 0$ (respectively $f^{\prime\prime} \ge 0$) if and only if there exist a non-decreasing (respectively, convex) function $\ft$ satisfying \eqref{eq:fn_dist}.\qed
\end{lemma}

We end this appendix by defining the convolution $f\conv g$ of distributions $f\in\Dc'$ and $g\in\Ec'$ (note that at least one distribution must have compact support). Recall that if $f\in L^1_\text{loc}$ and $\varphi\in\Cc^\infty_0$, we have $[f\conv \varphi](s):= \int_\Rb f(t)\varphi(s-t)\dd t$, $s\in\Rb$, such that $f\conv\varphi \in\Cc^{\infty}$. For $f\in\Dc'$ and $\varphi\in\Cc^{\infty}_0$, we thus define the convolution $f\conv \varphi:\Rb\to\Cb$ as the function $s\mapsto [f\conv\varphi](s):=f\big(\varphi(s-\cdot)\big)$, $s\in\Rb$.
The convolution of a distribution and a $\Cc^{\infty}_0$ function has the following properties.
\begin{lemma}\label{lem:Conv_Dist1}
	(Theorems 4.1.1-2 on page 88 of \cite{Hormander1990}) 
	For all $f\in\Dc'$ and $\varphi,\psi \in C^\infty_0(\mathbb{R})$, we have that: 
	$(i)$ $f\conv\varphi\in C^\infty$; 
	$(ii)$ $\supp(f\conv\varphi)\subseteq \supp(f)+\supp(\varphi):=\big\{x+y:x\in\supp(f),\,y\in\supp(\varphi)\big\}$; 
	$(iii)$ $(f\conv\varphi)^\prime = f^\prime \conv \varphi = f\conv \varphi^\prime$; and 
	$(iv)$ $f\conv(\varphi\conv\psi) = (f\conv\varphi)\conv\psi$.\qed
\end{lemma}
For $a\in\Rb$, define the \emph{translation operator} $\tau_a:\Cc^{\infty}_0\to\Cc^{\infty}_0$ by
\begin{align}\label{eq:translation}
	\tau_a(\varphi) := [\delta_a\conv\varphi]=\varphi(\cdot-a);\quad \varphi\in\Cc^\infty_0,
\end{align}
in which $\delta_a$ is the Dirac measure given by \eqref{eq:DeltaFn}. Direct calculation shows that  $f\conv\tau_a(\varphi) = \tau_a(f\conv \varphi)$ for all $\varphi\in\Cc^\infty_0$. As the following result indicates, the converse of this statement is also true, that is, the only continuous linear map that commutes with all translations is convolution.
\begin{lemma}\label{lem:ConvTrans}
	(Theorems 4.2.1 on page 100 of \cite{Hormander1990}) 
	Consider a continuous linear map $\Lc:C^\infty_0\to C^\infty$ such that $\Lc(\varphi_n)\to 0$ in $\Cc^\infty$ for all $\varphi_n\rightarrow 0$ in $\Cc^\infty_0$. If $\Lc$ commutes with all translations (i.e. $\Lc\big(\tau_a(\varphi)\big) = \tau_a\big(\Lc(\varphi)\big)$ for all $a\in\Rb$ and $\varphi\in\Cc^\infty_0$), then there exists a unique $f\in\Dc'$ such that $\Lc(\varphi) = f\conv\varphi$ for all $\varphi\in \Cc^\infty_0$.\qed 
\end{lemma}

The previous lemma has the following important consequence.
\begin{corollary}
	Let $f_1\in\Dc'$ and $f_2\in\Ec'$. Then, there exists a unique distribution $f\in\Dc'$ such that $f_1\conv(f_2\conv \varphi) = f\conv \varphi$ for all $\varphi\in\Cc^\infty_0$.\qed
\end{corollary}
\begin{proof}
	Apply Lemma \ref{lem:ConvTrans} to $\Lc(\varphi):=f_1\conv(f_2\conv \varphi)$, $\varphi\in\Cc^\infty_0$. Note that we need $f_2\in\Ec'$ so that $\supp(f_2\conv \varphi)$ remains a compact set by Lemma \ref{lem:Conv_Dist1}.$(ii)$.
\end{proof}

Finally, by exploiting the previous lemma, we define convolution of two distributions one of which has a compact support.
\begin{definition}\label{def:ConvDist}
	For $f_1\in\Dc'$ and $f_2\in\Ec'$, we define $f_1\conv f_2=f_2\conv f_1:= f$, in which $f\in\Dc'$ is the unique distribution satisfying $f_1\conv(f_2\conv \varphi) = f\conv \varphi$ for all $\varphi\in\Cc^\infty_0$.\qed
\end{definition}

%
%

\subsection{Tempered distributions and their Fourier transform}\label{sec:Fourier}

The Fourier transform of a function $f\in L^1$ is given by $\F[f](s):=\int_\Rb \ee^{-\ii st}f(t)\dd t$, $s\in\Rb$. The goal of this appendix is to define the Fourier transform for a special class of distributions called tempered distributions, and explore the properties of this generalization of the Fourier transform. 

Let $\Sc$ be the space of \emph{rapidly decreasing functions} (also know as \emph{Schwartz's functions}), which is the set of all functions $\phi\in\Cc^{\infty}$ that satisfy
\begin{align}\label{eq:S_norm}
	\|\varphi\|'_{n,m} := \sup_{t\in\Rb}\left|t^m\frac{\dd^n \varphi(t)}{\dd t^n}\right|<+\infty,
\end{align}
for all $n,m\in\Nb_0$. If $\varphi\in\Sc$, then \eqref{eq:S_norm} yields that, $\frac{d^n\varphi(t)}{dt^n}\to 0$ as $x\to\pm\infty$ for all $n\ge0$. In other words, a rapidly decreasing function has vanishing derivatives of all orders. It can be shown that $\Cc^{\infty}_0 \subset \Sc \subset \Cc^{\infty}$, that $\Sc\subset L^1$, and that $C^{\infty}_0$ is dense in $\Sc$.
Since $\Sc\subset L^1$, any $\varphi\in\Sc$ has the classical Fourier transform $\F[\varphi](s):=\int_\Rb\ee^{\ii st}\varphi(t)\dd t$, $s\in\Rb$. Indeed, the significance of the space $\Sc$ is that the classical Fourier transform is an isomorphism $\F:\Sc\to\Sc$.
\begin{lemma}\label{lem:Isomorphism_Schwartz}
	(Theorem 7.1.5 on page 161 of \cite{Hormander1990})
	For any $\varphi\in\Sc$, the Fourier transform $\F[\varphi](s):=\int_\Rb\ee^{-\ii st}\varphi(t)\dd t$, $s\in\Rb$, is in $\Sc$. Furthermore, the map $\varphi\mapsto \F[\varphi]:\Sc\to\Sc$ is a linear continuous map with a linear continuous inverse given by Fourier's inversion formula
	\begin{align}\label{eq:InvFourier}
		\F^{-1}[\varphi](t) := \frac{1}{2\pi} \int_\Rb \ee^{\ii st} \varphi(s) \dd s;\quad t\in\Rb,
	\end{align}
	or, equivalently, by $\F^2[\varphi](s) = 2\pi\varphi(-s)$, $s\in\Rb$ and $\varphi\in\Sc$.\qed
\end{lemma}

As mentioned earlier, our goal is to define the Fourier transform for tempered distributions, which we define next.
\begin{definition}\label{def:TempDist}
	The space of \emph{tempered distributions} $\Sc'$ is the dual space of $\Sc$ with the topology generated by the seminorms $\|\cdot\|'_{n,m}$ given by \eqref{eq:S_norm}. In other words, a linear functional $f:\Sc\to\Cb$ is a tempered distribution if there exists a $k\in\Nb_0$ and a constant $C>0$ such that $\big|f(\varphi)\big|\le C \max\big\{\|\varphi\|'_{m,n}:m+n\le k\big\}$ for all $\varphi\in\Sc$. Note that $\Sc' \subset \Dc'$ since $\Cc_0^\infty\subset \Sc$.\qed
\end{definition}

Assume that $f\in\Sc'$ is a function, that is
\begin{align}
	f(\varphi)=\int_\Rb \ft(t)\varphi(t)\dd t;\quad \varphi\in\Sc,
\end{align}
for some function $\ft:\Rb\to\Rb$ (recall \eqref{eq:fn_dist}). Since the improper integral on the right side of the above equation needs to converge for all rapidly decreasing functions $\varphi$, it follows that $\ft$ cannot grow too fast at $\pm\infty$. Thus, we may refer to $\Sc'$ as the space of ``slowly growing distributions.'' In particular, $\Ec'\subset \Sc'$ since $\Sc\subset \Cc^\infty$ (recall, from Definition \ref{def:compdist}, that $\Ec'$ is the set of distributions with compact support). Other elements of $\Sc'$ are functions $f$ with polynomial growth (i.e. $|f(t)|\le C(1+|t|)^m$, $t\in\Rb$, for some $C,m>0$) and measures $\mu$ on $\Rb$ satisfying $\int_\Rb (1+|t|)^{-m}\mu(\dd t)<+\infty$ for some $m>0$. This implies that the product of a polynomial with an $L^p$ function, $p\ge1$, is a tempered distribution. Finally, $\Sc'$ is closed under differentiation and multiplication by elements of $\Sc$ (as defined by Definition \ref{def:DerivProdDist}).

By exploiting Lemma \ref{lem:Isomorphism_Schwartz}, we now define the Fourier transform $\F[f]$ of a tempered distributions $f\in\Sc'$.

\begin{definition}\label{def:Fourier_tempered}
	For $f\in\Sc'$, define $\F[f]\in\Sc'$ by $\F[f](\varphi) := f\big(\F[\varphi]\big)$, $\varphi\in\Sc$. That $\F[f]\in\Sc'$ follows from Lemma \ref{lem:Isomorphism_Schwartz}. \qed
\end{definition}

For all $f\in\Sc'$ and $\varphi\in\Sc$, applying the above definition twice and then Lemma \ref{lem:Isomorphism_Schwartz} yields Fourier's inversion formula for tempered distributions,
\begin{align}
	\F^2[f](\varphi) := \F\big[\F[f]\big](\varphi)=\F[f]\big(\F[\varphi]\big)
	=f\big(\F^2(\varphi)\big)=2\pi f\big(\check{\varphi}\big),
\end{align}
in which we have defined $\check{\varphi}(t) := \varphi(-t)$. In fact, the counterpart of Lemma \ref{lem:Isomorphism_Schwartz} also holds for the Fourier transform on the space of tempered distributions. That is, for all tempered distributions, the inverse Fourier transform exists and is continuous.

\begin{lemma}\label{lem:Isomorphism_Tempered}
	(Theorem 7.1.10 on page 164 of \cite{Hormander1990})
	The Fourier transform, given by Definition \ref{def:Fourier_tempered}, is a continuous linear map $f\mapsto \F[f]:\Sc'\to\Sc'$ (with the weak topology of $\Sc'$). Furthermore, it has a linear continuous inverse $f\mapsto \F^{-1}[f]:\Sc'\to\Sc'$ given by 
	$\F^{-1}[f] := \frac{1}{2\pi} \F[\check{f}]$ 
	or, equivalently, by Fourier's inversion formula
	$\F^2[f] = 2\pi \check{f}$, 
	in which $\check{f}$ is given by 
	$\check{f}(\varphi) = f\big(\varphi(-\cdot)\big)$, $\varphi\in\Sc$.\qed
\end{lemma}

\begin{example}\label{exm:Fourier_delta}
	Consider the Dirac measure $\delta_a$, $a\in\mathbb{R}$, given by \eqref{eq:DeltaFn}. Since $\del_a\in\Ec'\subset\Sc'$, $\F[\del_a]$ is defined as a tempered distribution. In fact, we have that
	\begin{align}
		\F[\del_a](\varphi):=\del_a\big(\F[\varphi]\big) = \F[\varphi](a) = \int_\mathbb{R} e^{-\ii a t} \varphi(t)dt,\quad \varphi\in\Sc.
	\end{align}
	In light of \eqref{eq:fn_dist}, $\F[\delta_a]$ is the tempered distribution corresponding to $\ft(t)=\ee^{-\ii a t}$, $t\in\Rb$, and, with the usual abuse of notation, we may write $\F[\del_a]=\ee^{-\ii a(\cdot)}$.	
	The Fourier's inversion formula (Lemma \ref{lem:Isomorphism_Tempered}) then yields 
	$\F[\ee^{\ii a(\cdot)}] = \F\big[\F[\del_{-a}]\big]=\F^2[\del_{-a}] = 2\pi\check{\del}_{-a} = 2\pi\del_a$.
	Equivalently, $\F^{-1}[\del_a] = \frac{1}{2\pi}\F[\del_{-a}]=\frac{1}{2\pi}\ee^{\ii a(\cdot)}$. By setting $a=0$, we obtain that $\F[\del_0]=1$ and $\F[1]=2\pi\del_0$.\qed
\end{example}

In the previous example, $\del_a\in\Ec'$ and $\F[\del_a](s)=\ee^{-i a s}=\del_a\big(\ee^{-is(\cdot)}\big)$ is an entire analytical function for $s\in\Cb$. As the following result shows, the same is true for the Fourier transform of any distribution with compact support.
\begin{lemma}\label{lem:FourierLaplace}
	(Theorem 7.1.14 on page 165 of \cite{Hormander1990})
	For all $f\in\Ec'$, $\F[f](\varphi)=\int_\Rb \fh(s)\varphi(s)\dd s$, $\varphi\in\Sc$, in which 
	$\fh(s):=f\big(\ee^{-is(\cdot)}\big)$. Therefore, with the usual abuse of notation, we may write $\F[f]= \fh$ for all $f\in\Ec'$. Furthermore, the function $\fh(s)$ is defined for all complex numbers $s\in\Cb$ and is an entire analytic function called the Fourier-Laplace transform of $f$.\qed
\end{lemma}

The following result states the convolution theorem for the Fourier transform of tempered distributions. Recall from Definition \ref{def:ConvDist} that, for the convolution $f_1\conv f_2$ to be defined, either $f_1$ or $f_2$ must have compact support.
\begin{lemma}\label{lem:Conv}
	(Theorem 7.1.15 on page 166 of \cite{Hormander1990})
	Assume that $f_1\in\Ec'$ and $f_2\in\Sc'$. Then $f_1\conv f_2\in\Sc'$ and $\F[f_1\conv f_2]=\F[f_1] \F[f_2]$. Note that the product on the right side is defined by Definition \ref{def:DerivProdDist} because $\F[f_1]\in\Cc^{\infty}$ by Lemma \ref{lem:FourierLaplace}.\qed
\end{lemma}

From Definitions \ref{def:DerivProdDist} and \ref{def:ConvDist}, we have $f'=f\conv\del'_0$, in which $\del'_0\in\Ec'$ is the derivative of the Dirac measure, that is, $\del'_0(\varphi)=\varphi'(0)$, $\varphi\in\Cc^{\infty}$. Using this fact, we obtain the following corollary of Lemma \ref{lem:Conv}.

\begin{lemma}\label{lem:Fourier_Deriv}
	For $f\in\Sc'$, we have that $\F[f']= isf$ and $\F[t f]= i\F[f]'$.\qed
\end{lemma}

%

%
%

\section{Proof of Proposition \ref{prop:ExistenceUniquness}}\label{app:ExistenceUniquness}

Take an arbitrary $\thetav\in\Xi$. Condition $(ii)$ yields that
\begin{align}\label{eq:J1k}
	\F[J_{1,k}(\cdot,\thetav)]\F[\mu_{\thetav,k}]=\F[J_{0,k}];\quad k\in\{1,2\}.
\end{align}
Note that the product on the left side of this equation is defined by Definition \ref{def:DerivProdDist} since $\F[\mu_{\thetav,k}]\in\Cc^{\infty}$ by Assumption \ref{assum:gam1gam2}. Next, we apply the convolution theorem (i.e. Lemma \ref{lem:Conv}) to the left side of \eqref{eq:J1k}. Doing so, however, require $\mu_{\thetav,k}$ to have compact support which is not true in general. For instance, $\supp(\mu_{\thetav,k})=\Rb$ in the Black-Scholes model, see Subsection \ref{sec:BS}. To circumvent this obstacle, we use a localization argument.

For $a>0$ and $k\in\{1,2\}$, define $\mut_{\thetav,k}^a(B):=\mu_{\thetav,k}\big(B\cap[-a,a]\big)$, $B\in\Bc(\Rb)$. Note that $\mut_{\thetav,k}^a\in\Ec'$ since $\supp(\mut_{\thetav,k}^a)\subseteq[-a,a]$. Furthermore,
$[J_{1,k}(\cdot,\thetav)\conv\mut_{\thetav,k}^a](s)=\int_\Rb J_{1,k}(s-t) \dd\mut_{\thetav,k}^a(t) = \int_{-a}^a J_{1,k}(s-t)\dd\mu_{\thetav,k}(t)$, $s\in\Rb$. Therefore, for $s\in\Rb$, we have that
\begin{align}
	\int_{-a}^a J_{1,k}(s-t)\dd\mu_{\thetav,k}(t) 
	= \F^{-1}\Big[\F\big[J_{1,k}(\cdot,\thetav)\conv\mut_{\thetav,k}^a\big]\Big](s)
	=\F^{-1}\Big[\F[J_{1,k}(\cdot,\thetav)]\F[\mut_{\thetav,k}^a]\Big](s),
\end{align}
in which the last step follows from the convolution theorem (i.e. Lemma \ref{lem:Conv}). Since $\F$ is an isomorphism of $\Sc'$ (by Lemma \ref{lem:Isomorphism_Tempered}), letting $a\to+\infty$ and using Condition $(iii)$ yield that
\begin{align}\label{eq:J1k_conv}
	\int_\Rb J_{1,k}(s-t)\dd\mu_{\thetav,k}(t) 
	= \F^{-1}\Big[\F[J_{1,k}(\cdot,\thetav)]\F[\mut_{\thetav,k}]\Big](s)
	=J_{0,k}(s);\quad s\in\Rb,
\end{align}
in which the last step follows from \eqref{eq:J1k}. Let  $J_1(t,\thetav):=\ee^{-\frac{1}{\gam_1}t}J_{1,1}(t,\thetav) + \ee^{-\frac{1}{\gam_2}t}J_{1,2}(t,\thetav)$, $t\in\Rb$, as in the statement of the proposition. For all $s\in\Rb$, we calculate
\begin{align}
	\int_\Rb J_1(s-t,\thetav)\dd\nut_\thetav(t)
	&= \int_\Rb \ee^{-\frac{1}{\gam_1}(s-t)}J_{1,1}(s-t,\thetav)\dd\nut_\thetav(t)
	+ \int_\Rb \ee^{-\frac{1}{\gam_2}(s-t)}J_{1,2}(s-t,\thetav)\dd\nut_\thetav(t)\\
	&= \ee^{-\frac{1}{\gam_1}s} \int_\Rb J_{1,1}(s-t,\thetav) \dd\mu_{\thetav,1}(t)
	+ \ee^{-\frac{1}{\gam_2}s} \int_\Rb J_{1,2}(s-t,\thetav) \dd\mu_{\thetav,2}(t)\\
	&= \ee^{-\frac{1}{\gam_1}s} J_{0,1}(s) + \ee^{-\frac{1}{\gam_2}s} J_{0,2}(s)
	= J_0(s),
\end{align}
in which the last three steps follow from \eqref{eq:mu_thetav}, \eqref{eq:J1k_conv}, and the definition of $J_{0,k}$. Since $\thetav$ was chosen arbitrarily, we conclude that $J_1$ solves the deconvolution problem \eqref{eq:Deconvolution}. That $J_1(\cdot,\thetav)\in\Jc(\gam_1,\gam_2)$ follows from
\begin{align}
	&\left(\ee^{\frac{1}{\gam_1}t}\1_{\{t<0\}}
	+\ee^{\frac{1}{\gam_2}t}\1_{\{t\ge0\}}\right)|J_1(t,\thetav)|
	\\&
	\le \1_{\{t<0\}}\left(|J_{1,1}(t,\thetav)| + \ee^{\left(\frac{1}{\gam_1}-\frac{1}{\gam_2}\right)t}|J_{1,2}(t,\thetav)|\right)
	+\1_{\{t\ge0\}}\left( \ee^{\left(\frac{1}{\gam_2}-\frac{1}{\gam_1}\right)t}|J_{1,1}(t,\thetav)| + |J_{1,2}(t,\thetav)|\right),
\end{align}
and that $J_{1,k}(\cdot,\thetav)$, $k\in\{1,2\}$, satisfies Condition $(ii)$.

To show the last statement of Proposition \ref{prop:ExistenceUniquness}, take an arbitrary $\thetav\in\Xi$. Assume that $\F[\mu_{\thetav,k}](\xi)\ne0$, $(k,\xi)\in\{1,2\}\times\Rb$, and that there exists $\Jt\in\Jc(\gam_1,\gam_2)$ satisfying $\int_\Rb\ee^{\frac{s-t}{\gam_k}}\Jt(s-t)\dd\mu_{\thetav,k}(t) = J_{0,k}(s)$, $(k,s)\in\{1,2\}\times\Rb$. We want to show that $\Jt=J_1(\cdot,\thetav)$ almost everywhere on $\Rb$. For $(k,t)\in\{1,2\}\times\Rb$, define $h_k(t) := J_{1,k}(t,\thetav)-\ee^{\frac{t}{\gam_k}}\Jt(t)$, $t\in\Rb$. Since $J_1(t,\thetav)-\Jt(t)=\ee^{-\frac{t}{\gam_1}}h_1(t)+\ee^{-\frac{t}{\gam_2}}h_2(t)$, $t\in\Rb$, it suffices to show that $\F[h_k]=0$, $k\in\{1,2\}$. By \eqref{eq:J1k_conv}, we have
\begin{align}
	\int_\Rb h_k(s-t)\dd\mu_{\thetav,k}(t)
	=\int_\Rb J_{1,k}(s-t,\thetav)\dd\mu_{\thetav,1}(t)
	-\int_\Rb\ee^{\frac{s-t}{\gam_k}}\Jt(s-t)\dd\mu_{\thetav,1}(t)=0,
\end{align}
for $(s,k)\in\Rb\times\{1,2\}$. Through a localization argument similar to the one used in the first part of the proof we obtain that $\F[h_k]\F[\mu_{\thetav,k}] = 0$. Since we have assumed that $\F[\mu_{\thetav,k}]\in\Cc^{\infty}$ and that $\F[\mu_{\thetav,k}](s)\ne0$, $s\in\Rb$, it follows that $\F[h_k]=0$, as we set out to prove.

%
%

\section{Proof of Theorem \ref{thm:CMIM}}\label{sec:CMIM_proof}
We first show that $I_1$ in \eqref{eq:IntEq_CMIM_sol} is well-defined and that $I_1(\cdot,\thetav)\in\CMIM(\gam_1,\gam_2)$, $\thetav\in\Xi$. For $\thetav\in\Xi$ and $\gam_1<\gam<\gam_2$, we have
\begin{align}
	\int_{0^+}^{1^-} \rho^{1-\frac{1}{\gam_2}}\dd\nu_\thetav(\rho)
	+\int_{1}^{+\infty} \rho^{1-\frac{1}{\gam_1}}\dd\nu_\thetav(\rho)
	\le	\int_{\Rb_+} \rho^{1-\frac{1}{\gam}}\dd\nu_\thetav(\rho)
	\le\int_{0^+}^{1^-} \rho^{1-\frac{1}{\gam_1}}\dd\nu_\thetav(\rho)
	+\int_1^{+\infty} \rho^{1-\frac{1}{\gam_2}}\dd\nu_\thetav(\rho).
\end{align}
From \eqref{eq:CMIM_nu}, we then obtain
\begin{align}\label{eq:CMIM_integ_sol}
	0<\eps_\thetav<\left(\int_{\Rb_+} \rho^{1-\frac{1}{\gam}}\dd\nu_\thetav(\rho)\right)^{-1}<M_\thetav<+\infty;\quad (\gam,\thetav)\in(\gam_1,\gam_2)\times\Xi),
\end{align}
in which
\begin{align}
	\eps_\thetav := \left(
	\int_{0^+}^{1^-} \rho^{1-\frac{1}{\gam_1}}\dd\nu_\thetav(\rho)
	+\int_1^{+\infty} \rho^{1-\frac{1}{\gam_2}}\dd\nu_\thetav(\rho)
	\right)^{-1},
	\intertext{and}
	M_\thetav :=\left(
	\int_{0^+}^{1^-} \rho^{1-\frac{1}{\gam_2}}\dd\nu_\thetav(\rho)
	+\int_{1}^{+\infty} \rho^{1-\frac{1}{\gam_1}}\dd\nu_\thetav(\rho)
	\right)^{-1}.
\end{align}
By \eqref{eq:CMIM_integ_sol}, the integral on the right side of \eqref{eq:IntEq_CMIM_sol} is convergent and, in particular,
\begin{align}
	\eps_\thetav I_0(y)\le I_1(y,\thetav)\le M_\thetav I_0(y);\quad (y,\thetav)\in\Rb_+\times\Xi.
\end{align}
Therefore, $I_1(\cdot,\thetav)\in\CMIM(\gam_1,\gam_2)$, $\thetav\in\Xi$.

Next, we show that $I_1$ is a solution of Problem \ref{prob:IntEq}. Since $\CMIM(\gam_1,\gam_2)\subset\Ic$ by Lemma \ref{lem:CMIM}, we have that $I_1(\cdot,\thetav)\in\Ic$, $\thetav\in\Xi$. That $\int_{\Rb_+} I_1(y\rho, \thetav)\dd \nu_\thetav(\rho)<+\infty$, $(y,\thetav)\in\Rb_+\times\Xi$, is shown as follows 
\begin{align}
	&\int_{\Rb_+} I_1(y\rho, \thetav)\dd \nu_\thetav(\rho)
	=
	\int_{\Rb_+} \left(
	\int_{\gam_1}^{\gam_2} (\rho y)^{-\frac{1}{\gam}}\left(
	\int_{\Rb_+} \rho^{1-\frac{1}{\gam}}\dd\nu_\thetav(\rho)
	\right)^{-1}\dd m_0(\gam)
	\right)\dd \nu_\thetav(\rho)\\
	&=\int_{\gam_1}^{\gam_2} y^{-\frac{1}{\gam}} \left(
	\frac{\int_{\Rb_+}\rho^{-\frac{1}{\gam}}\dd \nu_\thetav(\rho)}
	{\int_{\Rb_+} \rho^{1-\frac{1}{\gam}}\dd\nu_\thetav(\rho)}
	\right)\dd m_0(\gam)
	\le\int_{\gam_1}^{\gam_2} y^{-\frac{1}{\gam}} \left(
	\int_{0^+}^{1^-}\rho^{-\frac{1}{\gam_1}}\dd \nu_\thetav(\rho)+\nu_\thetav\big([1,+\infty)\big)
	\right)
	M_\thetav \dd m_0(\gam)\\
	&= \left(
	\int_{0^+}^{1^-}\rho^{-\frac{1}{\gam_1}}\dd \nu_\thetav(\rho)+\nu_\thetav\big([1,+\infty)\big)
	\right)
	M_\thetav I_0(y)<+\infty,
\end{align}
in which the third step follows from \eqref{eq:CMIM_integ_sol} and the last step uses \eqref{eq:CMIM_nu}.
Finally, $I_1$ satisfies \eqref{eq:IntEq} since
\begin{align}
	\int_{\Rb_+} \rho I_1(y\rho, \thetav)\dd \nu_\thetav(\rho)
	&=\int_{\Rb_+} \rho \left(
	\int_{\gam_1}^{\gam_2} (\rho y)^{-\frac{1}{\gam}}\left(
	\int_{\Rb_+} \rho^{1-\frac{1}{\gam}}\dd\nu_\thetav(\rho)
	\right)^{-1}\dd m_0(\gam)
	\right)\dd \nu_\thetav(\rho)\\
	&=\int_{\gam_1}^{\gam_2} y^{-\frac{1}{\gam}} \left(
	\frac{\int_{\Rb_+}\rho^{1-\frac{1}{\gam}}\dd \nu_\thetav(\rho)}
	{\int_{\Rb_+} \rho^{1-\frac{1}{\gam}}\dd\nu_\thetav(\rho)}
	\right)\dd m_0(\gam)
	=I_0(y),
\end{align}
for all $(y,\thetav)\in\Rb_+\times\Xi$.	We have shown that $I_1$ is a solution of Problem \ref{prob:IntEq}.

Finally, we show that $I_1$ is the only solution such that $I_1(\cdot,\thetav)\in\CMIM(\gam_1,\gam_2)$, $\thetav\in\Xi$. Let \begin{align}\label{eq:It_CMIM}
	\It(y,\thetav):=\int_{\gam_1}^{\gam_2}y^{-\frac{1}{\gam}}\dd\mt_\thetav(\gam);\quad (y,\thetav)\in\Rb_+\times\Xi,
\end{align}
be a solution of Problem \ref{prob:IntEq}, in which $\mt_\thetav$, $\thetav\in\Xi$, are finite Borel measures with $\supp(\mt_\thetav)\subset(\gam_1,\gam_2)$. Take an arbitrary $\thetav\in\Xi$. Define the Borel measure $\mt_{0,\thetav}$ by
\begin{align}\label{eq:mu_0tht}
	\mt_{0,\thetav}(B) := \int_B \left(\int_{\Rb_+}\rho^{1-\frac{1}{\gam}}\dd \nu_\thetav(\rho)\right) \dd \mt_\thetav(\gam);\quad B\in\Bc(\Rb).
\end{align}
By \eqref{eq:CMIM_integ_sol}, $\mt_\thetav$ and $\mt_{0,\thetav}$ are equivalent, $\supp(\mt_{0,\thetav})=\supp(\mt)\subset(\gam_1,\gam_2)$, and we have that
\begin{align}\label{eq:mu_tht}
	\mt_\thetav(B) = \int_B \left(\int_{\Rb_+}\rho^{1-\frac{1}{\gam}}\dd \nu_\thetav(\rho)\right)^{-1} \dd \mt_{0,\thetav}(\gam);\quad B\in\Bc(\Rb).
\end{align}
From \eqref{eq:mu_0tht} and since $\It$ solves \eqref{eq:IntEq}, it follows that
\begin{align}
	\int_{\Rb_+} y^{-\frac{1}{\gam}} \dd \mt_{0,\thetav}(\gam)
	=\int_{\gam_1}^{\gam_2} y^{-\frac{1}{\gam}} \left(\int_{\Rb_+}\rho^{1-\frac{1}{\gam}}\dd \nu_\thetav(\rho)\right) \dd \mt_\thetav(\gam)
	=I_0(y)
	= \int_{\Rb_+} y^{-\frac{1}{\gam}} \dd m_0(\gam),\quad y>0.
\end{align}
The above equation implies that $\mt_{0,\thetav}$ and $m_0$ have the same Laplace–Stieltjes transform and, therefore, $\mt_{0,\thetav}=m_0$. From \eqref{eq:IntEq_CMIM_sol}, \eqref{eq:It_CMIM}, and \eqref{eq:mu_tht}, we obtain that $\It(\cdot,\thetav)=I_1(\cdot,\thetav)$ for all $\thetav\in\Xi$.

\end{document}